\documentclass[11pt]{article}

\usepackage[centering,margin=1.5cm]{geometry}

\usepackage[utf8]{inputenc}
\usepackage{blindtext}
\usepackage{setspace}
\usepackage{graphicx}
\usepackage{notoccite} 
\usepackage{lscape} 
\usepackage{caption} 
\usepackage{amsmath,amssymb,amsfonts,amsthm}
\usepackage{comment}
\usepackage{listings}
\usepackage{verbatim}
\usepackage{xcolor}
\usepackage{algorithm}
\usepackage{algpseudocode}
\usepackage{textcomp}
\usepackage{enumitem}
\usepackage{stmaryrd}
\usepackage{cases,xspace}
\usepackage{hyperref}
\usepackage{subcaption}
\usepackage{multirow}
\usepackage{multicol}
\usepackage{mdframed}

\usepackage{tikz}    
\usetikzlibrary{matrix, positioning, fit, calc}
\usepackage{version}
\usepackage[rightcaption]{sidecap}

\usepackage{tikz}
\usetikzlibrary{positioning}

\DeclareCaptionLabelFormat{figure}{{Figure #2}}
\usepackage[prependcaption,colorinlistoftodos]{todonotes}

\usepackage{cancel} 
\definecolor{codeblue}{rgb}{0.2,0.2,0.7}
\definecolor{codegreen}{rgb}{0,0.6,0}
\definecolor{codegray}{rgb}{0.5,0.5,0.5}
\definecolor{codepurple}{rgb}{0.58,0,0.82}

\lstdefinestyle{python}{
    language=Python,
    basicstyle=\ttfamily\footnotesize,
    keywordstyle=\color{codeblue}\bfseries,
    stringstyle=\color{codepurple},
    commentstyle=\color{codegreen}\itshape,
    numbers=left,
    numberstyle=\tiny\color{codegray},
    stepnumber=1,
    numbersep=10pt,
    backgroundcolor=\color{white},
    frame=single,
    breaklines=true,
    breakatwhitespace=false,
    showspaces=false,
    showstringspaces=false,
    showtabs=false,
    tabsize=4,
    captionpos=b
}
\usepackage{empheq}

\newtheorem{definition}{Definition}
\newtheorem{thm}{Theorem}
\newtheorem{lem}{Lemma}
\newtheorem{pro}{Proposition}

\newtheorem{remark}{Remark}
\newtheorem{conj}{Conjecture}

\newmdtheoremenv{assumptionbox}{Assumption}
\newtheorem*{metatheorem}{Informal statement}
\newenvironment{rem}{\begin{remark}}{\qed\end{remark}}
\newcommand{\bd}{\begin{definition}} 
\newcommand{\ed}{\end{definition}} 
\newcommand{\bp}{\begin{pro}} 
\newcommand{\ep}{\end{pro}} 
\newcommand{\bt}{\begin{thm}} 
\newcommand{\et}{\end{thm}}
\newcommand{\blm}{\begin{lem}} 
\newcommand{\elm}{\end{lem}}
\let\eps\varepsilon

\newcommand{\R}{\mathbb{R}}

\newcommand{\mcR}{\mathcal{R}}

\DeclareMathOperator{\rank}{\operatorname{rank}}

\newcommand{\bi}{\begin{itemize}} 
\newcommand{\ei}{\end{itemize}} 
\newcommand{\bds}{\begin{description}} 
\newcommand{\eds}{\end{description}} 
\newcommand{\beq}{\begin{equation}} 
\newcommand{\eeq}{\end{equation}} 

\newcommand{\vecto}{\operatorname{vec}}
\newcommand{\tr}{\operatorname{tr}}

\newcommand{\norm}[1]{\|#1\|}

\newcommand{\normF}[1]{\|#1\|_\textup{F}}
\newcommand{\wnormF}[2]{\|#1\|_\textup{F,#2}}
\newcommand{\inprodF}[1]{\langle#1\rangle_\textup{F}}

\newcommand{\der}[2]{\frac{d #1}{d #2}} 

\newcommand{\trasp}[1]{#1^{\top}}
\newcommand{\diag}[1]{\operatorname{diag}(#1)}

\newcommand{\Bignorm}[1]{\Big\|#1\Big\|}

\newcommand{\until}[1]{\{1,\dots, #1\}}

\newcommand{\subscr}[2]{#1_{\textup{#2}}}

\newcommand{\setdef}[2]{\{#1 \; | \; #2\}}

\newcommand{\bigsetdef}[2]{\big\{#1 \; | \; #2\big\}}

\newcommand{\map}[3]{#1\colon #2 \rightarrow #3}

\DeclareMathOperator*{\argmin}{arg\,min}
\DeclareMathOperator*{\argmax}{arg\,max}

\newcommand{\bigO}{\mathcal{O}}
\newcommand{\ds}{\displaystyle}

\newcommand{\e}{\mathrm{e}}

\newcommand{\0}{\mbox{\fontencoding{U}\fontfamily{bbold}\selectfont0}}

\renewcommand{\S}{\bold{S}}

\definecolor{gnred}{RGB}{255,91,89}

\definecolor{gnred1}{RGB}{71,0,0} 
\definecolor{gnred2}{RGB}{117,0,0} 
\definecolor{gnred3}{RGB}{164,0,0} 
\definecolor{gnred4}{RGB}{211,0,0} 
\definecolor{gnred5}{RGB}{255,0,0} 
\definecolor{gnred6}{RGB}{255,42,34} 
\definecolor{gnred7}{RGB}{255,91,89} 

\definecolor{gnblue1}{RGB}{0,36,71}   
\definecolor{gnblue2}{RGB}{0,60,118}  
\definecolor{gnblue3}{RGB}{0,85,164}
\definecolor{gnblue4}{RGB}{0,108,212}
\definecolor{gnblue4}{RGB}{0,108,212}
\definecolor{gnblue5}{RGB}{0,133,255}  
\definecolor{gnblue6}{RGB}{35,156,255} 
\definecolor{gnblue7}{RGB}{88,177,255} 

\definecolor{gnbrown1}{RGB}{71,27,0}  
\definecolor{gnbrown2}{RGB}{117,45,0} 
\definecolor{gnbrown3}{RGB}{164,62,0} 
\definecolor{gnbrown4}{RGB}{211,80,0} 
\definecolor{gnbrown5}{RGB}{255,97,0} 
\definecolor{gnbrown6}{RGB}{255,127,26} 
\definecolor{gnbrown7}{RGB}{255,155,86} 

\usepackage{hyperref}%
\hypersetup{colorlinks=true, linkcolor=blue, breaklinks=true, urlcolor=blue, citecolor=blue}

\title{Similarity Matching Networks: \\ Hebbian Learning
and Convergence Over Multiple Time Scales}

\author{
  Veronica Centorrino\thanks{Department of Information and Electric Engineering and Applied Mathematics, University of Salerno, Italy. {\tt\small \{vcentorrino, giovarusso\}@unisa.it}. VC and GR are supported by the European Union-Next Generation EU Mission 4 Component 1 CUP E53D23014640001.}
\and Francesco Bullo\thanks{Center for Control, Dynamical Systems, and Computation, UC Santa Barbara, CA, USA. {\tt\small bullo@ucsb.edu}. FB is supported in part by AFOSR grant FA9550-22-1-0059.}
\and Giovanni Russo\footnotemark[1]
}

\begin{document}
\maketitle
\begin{abstract}
A recent breakthrough in biologically-plausible normative frameworks for dimensionality reduction 
is based upon the similarity matching cost function and the low-rank matrix approximation problem.
Despite clear biological interpretation, successful application in several domains, and 
experimental validation, a formal complete convergence analysis remains elusive.
Building on this framework, we consider and analyze a continuous-time neural network, the \emph{similarity matching network}, for principal subspace projection. 
Derived from a min-max-min objective, this biologically-plausible network consists of three coupled dynamics evolving at different time scales:  neural dynamics, lateral synaptic dynamics, and feedforward synaptic dynamics at the fast, intermediate, and slow time scales, respectively. The feedforward and lateral synaptic dynamics consist of Hebbian and anti-Hebbian learning rules, respectively.
By leveraging a multilevel optimization framework, we prove convergence of the dynamics in the offline setting. 
Specifically, at the first level (fast time scale), we show strong convexity of the cost function and global exponential convergence of the corresponding gradient-flow dynamics. 
At the second level (intermediate time scale), we prove strong concavity of the cost function and exponential convergence of the corresponding gradient-flow dynamics within the space of positive definite matrices.
At the third and final level (slow time scale), we study a non-convex and non-smooth cost function, provide explicit expressions for its global minima, and prove almost sure convergence of the corresponding gradient-flow dynamics to the global minima. These results rely on two empirically motivated conjectures that are supported by thorough numerical experiments. Finally, we validate the effectiveness of our approach via a numerical example.
\end{abstract}

\section{Introduction}
\label{2024-sec:introduction}
Dimensionality reduction problems are ubiquitous in a wide range of domains spanning, e.g., neuroscience, signal processing, and machine learning~\cite{JW-YM:22}. Within this class of problems, which involve obtaining a lower-dimensional representation of a large input matrix while preserving the key features of the original data, here we consider principal subspace projections -- minimization problems with objectives that capture the discrepancy between the similarity of each pair of output vectors and the similarity of the corresponding input vectors. This optimization framework is known under various names, including low-rank matrix approximation~\cite{XL-ZW-YZ:15}, multidimensional scaling~\cite{TFC-MAC:00}, similarity matching~\cite{TH-CP-DBC:14}.

Over the past years, a growing body of theoretical and experimental evidence has supported the use of dimensionality reduction in neural systems~\cite{HBB:72, EO:82, BAO-DJF:96, LC-DYT:17} and this has ignited the quest for biologically plausible neural architectures capable of implementing this functionality.
A key challenge in neuroscience is to understand how such low-dimensional representations are formed in neural systems~\cite{DHH-TNW:68, HBB:72, BAO-DJF:96}, with far-reaching implications for, e.g., the design of artificial neural networks~\cite{AB-JR-CJR:12, WG-WMK-RN-LP:14, CP-DBC:19, JW-YM:22}. A promising approach to analyze/design neural networks is to establish a normative framework, see, e.g.,~\cite{BS-YB:17, CP-DBC:19, VC-AG-AD-GR-FB:23a}, which entails formulating biologically motivated optimization problems (OP) and transcribing these problems into biologically plausible networks able to compute the optimal solution of the original OP. 
In this context, a recent breakthrough in biologically-plausible normative frameworks for dimensionality reduction based upon the similarity matching cost function has been proposed in~\cite{TH-CP-DBC:14}, further developed in~\cite{CP-TH-DBC:15, CP-AMS-DBC:17} and subsequent works (see~\cite{CP-DBC:19} for a survey).
These networks derive from a cost function that embeds the similarity matching objective into a higher-dimensional space comprising neural, lateral synaptic, and feedforward synaptic variables, and implement local Hebbian and anti-Hebbian learning rules~\cite{DOH:49, WG-WK:02}. Additionally, the proposed network are closely related to the F\"{o}ldi{\'{a}}k's network~\cite{PF:90}, which features Hebbian updates for feedforward connections and anti-Hebbian updates for lateral connections. 
The framework proposed by~\cite{TH-CP-DBC:14, CP-TH-DBC:15, CP-AMS-DBC:17} -- which relies on discrete-time synaptic dynamics -- has been successfully applied to a number of dimensionality reduction problems~\cite{CP-DBC:14, CP-DC:15, CP-SM-DBC:17, DL-CP-DBC:22} and has been experimentally validated~\cite{NMC-CP-DBC:23}. However, convergence has been shown mostly empirically and a formal complete convergence analysis in offline and online settings remains elusive. Establishing convergence in the offline framework is not only key for theoretical validation but also provides a foundation for extending these guarantees to online computations, mimicking real-time neural computations.
Moreover, continuous-time modeling of both neural activity~\cite{JJH-DWT:86, MMC-JPC-at-all:12, SS-JPC:19} and synaptic plasticity~\cite{DWD-JJH:92, MG-OF-PB:12, LK-ML-JJES-EKM:20, VC-FB-GR:22k, DGC-LFA:24}, possibly co-evolving at different time scales, might hold an advantage in terms of biological plausibility over discrete-time networks. When taking this continuous-time modeling approach, tools from continuous-time dynamical systems can be used to characterize key network properties. For example, contraction theory can be used to characterize \emph{stability} and \emph{robustness}~\cite{WL-JJES:98, FB:24-CTDS, GR-MDB-EDS:10a, SX-GR-RHM:21}, also in dynamics with multiple time scales~\cite{LC-FB-EDA:23g}.

Motivated by these considerations, we consider and characterize the behavior of a continuous-time neural network inspired by that in~\cite{CP-AMS-DBC:17}, that features local Hebbian and anti-Hebbian learning rules. The network consists of three coupled dynamics -- neural, lateral synaptic, and feedforward synaptic -- co-evolving on fast, intermediate, and slow time scale, respectively. Specifically, our key technical contributions can be summarized as follows:
\begin{enumerate}
\item We introduce a modification of the network proposed by~\cite{CP-AMS-DBC:17}: the \emph{similarity matching network over three time scales}. This is a continuous-time neural network that naturally arises from the embedded similarity matching objective. This network consists of three dynamics co-evolving on three different time scales: neural dynamics, lateral synaptic dynamics, and feedforward synaptic dynamics at the fast, intermediate, and slow time scales, respectively. The feedforward and lateral synaptic dynamics consist of Hebbian and anti-Hebbian learning rules, respectively.
By leveraging a multilevel optimization framework, we prove convergence of the dynamics in the offline setting. In this framework, the OP from which the dynamics are derived is structured in three distinct levels, each corresponding to a specific time scale: (i) the first level addresses the OP with respect to the neural variables, (ii) the second level addresses the OP with respect to the lateral synaptic weights, and (iii) the third level focuses on the OP with respect to feedforward synaptic weights.

\item The three co-evolving dynamics follow from a three-level approach to solve the underlying OP. We thoroughly characterize the behavior of our proposed network and its interplay with the OP. Namely: (a) the first/fast level consists of a minimization problem with a strongly convex cost function. This ensures that a unique optimal solution for the OP, say $Y^\star$, exists -- we compute $Y^\star$ explicitly. At this level, we derive the dynamics for the neural variables, which we term \emph{continuous-time gradient-flow neural dynamics}. Leveraging contraction theory, we prove that this dynamics globally exponentially converge to $Y^\star$ -- that is, solutions of this dynamics exponentially converge to $Y^\star$ regardless of initial conditions. By proving contractivity, we also guarantee that the dynamics features a number of highly ordered transient and asymptotic behaviors; (b) the second/intermediate level consists of a maximization problem with concave cost. We prove strong concavity of the cost function, guaranteeing a unique optimal solution for the OP, say  $M^\star$, in the space of positive definite matrices. We provide the explicit expression of $M^\star$ and derive the dynamics of the lateral synapses. We show -- again using contraction tools -- that this \emph{continuous-time gradient-flow lateral synaptic dynamics} exponentially converges to $M^\star$;
(c) for the third/slow and final level, which involves a non-convex and non-smooth cost function and corresponds to feedforward synapses, we provide explicit expressions for its global minima, denoted by $W^\star$.
Then, we propose the corresponding \emph{continuous-time gradient-flow feedforward synaptic dynamics} and show that these dynamics converge to a global minimum $W^\star$ almost surely from random initialization. This result relies on two mild technical conjectures: (1) the forward invariance of the set of full-rank matrices -- remarkably, if this assumption is not satisfied, then the dynamics itself is not well-posed; and (2) a characterization of strict saddle points, which is standard in the literature although in this case analytically challenging to verify. Nevertheless, for both conjectures, we provide strong intuitive support and empirical validation in Appendix~\ref{apx:conjectures}. The code to replicate our numerical analysis is available at \url{https://shorturl.at/l1qVf}.
\item Moreover, for the continuous-time gradient-flow lateral synaptic dynamics, we show a key invariance property: the lateral weight matrix remains positive definite and symmetric if initialized so.
\item Finally, we illustrate the effectiveness of our results via numerical experiments.
\end{enumerate}

The rest of this paper is organized as follows. First, we provide the mathematical preliminaries (Section~\ref{sec:mathematical_preliminaries}) and (Section~\ref{sec:similarity_matching}) a self-contained review of the similarity matching problem together with the biologically plausible offline optimization algorithm from~\cite{CP-AMS-DBC:17}. In Section~\ref{sec:no_math}, we intuitively describe our technical {framework} and give the main results in Section~\ref{sec:convergence}. Proofs of the main and additional instrumental results, together with empirical evidence supporting our conjectures, are left to the appendices. In Section~\ref{sec:simulations}, we illustrate the effectiveness of our approach via a numerical example. A final discussion and concluding remarks are given in Section~\ref{sec:conclusions}. 
\section{Mathematical Preliminaries}
\label{sec:mathematical_preliminaries}
Given $x \in \R^n$, $\diag{x} \in \R^{n \times n}$ denotes the diagonal matrix with diagonal entries equal to $x$. $I_n$ and $\0_{n,m}$ are the $n \times n$ identity matrix and the $n \times m$ zero matrix, respectively. The symbol $\otimes$ denotes the Kronecker product. The set of $n \times n$ orthogonal matrices is $\mathcal{O}_n = \setdef{A \in \R^{n \times n}}{A^\top A = A A^\top = I_n}$. We let $\mathcal{P}_{n} \subset \mathcal{O}_n$ denote the set of $n \times n$ permutation matrices, that is the set of $n\times n$ binary matrix having exactly one entry equal to $1$ in each row and each column with all other entries equal to $0$. We let $\S^n$ (resp. $\S_{\succ 0}^n$) denote the set of real symmetric (resp. symmetric positive definite) $n\times n$ matrices. For $A, B \in \S^n$, we write $A \preceq B$ (resp. $A \prec B$) if $B-A$ is positive semidefinite (resp. definite). The vector of eigenvalues of a matrix $A \in \R^{n \times n}$ is denoted by $\lambda^A \in \R^n$, and $\subscr{\lambda}{min}(A)$ denotes the minimum eigenvalue of $A$. The trace of $A$ is $\tr(A) = \sum_{i=1}^{n} a_{ii}$, where $a_{ii}$ are the diagonal entries of $A$. A triple $(U_A, \Sigma_A, V_A)$ denotes the \emph{singular value decomposition} (SVD) of a matrix $A \in \R^{m \times n}$, that is $A = U_A \Sigma_A V_A^\top$, where 
\bi
\item $U_A = (u_1^A, \dots, u_m^A) \in \R^{m \times m}$, with $u_i^A \in \R^{m}$ being orthogonal and having the eigenvectors of $A A^\top$ as columns;
\item $V_A = (v_1^A, \dots, v_n^A) \in \R^{n \times n}$, with $v_j^A \in \R^n$ being orthogonal and having the eigenvectors of $A^\top A$ as columns;
\item the matrix $\Sigma_A$ is given by
\[
\Sigma_A =
\begin{bmatrix}
\Sigma^r_A & \0_{r,n-r} \\
\0_{m-r, r} & \0_{n-r,n-r} \\
\end{bmatrix} \in \R^{m \times n},
\]
with $\Sigma^r_A \in \R^{r \times r}$, $r \leq \min\{m,n\}$, being positive and diagonal. Namely, $\Sigma^r_A = \diag{\sigma_j^A} := \diag{\sqrt{\lambda_j}}$, $j \in \until{r}$, with $\lambda_j$ positive eigenvalues of $A A^\top$ and $A^\top A$.
\ei
When $m = n$, we have $U_A = V_A$, $\Sigma_A = \Lambda_A := \diag{\lambda^A}$, and refer to the SVD of $A$ by the pair $(U_A, \Sigma_A)$. Without loss of generality, the components of $\lambda^A$ are always considered ordered (i.e., $\lambda^A_1 \geq \dots \geq \lambda^A_n >0$).

Given $A$ and $B \in \R^{n\times m}$, their \emph{Frobenius inner product} is defined as $\inprodF{A, B} = \tr(\trasp{A}B)$. The corresponding  \emph{Frobenius norm} is given by $\normF{A}^2 := \tr (A\trasp{A}) = \tr (\trasp{A}A)$.
Given $P \in \S_{\succ 0}^n$, the \emph{weighted Frobenius norm} is $\wnormF{A}{P} = \tr(\trasp{A}PA)$.
The \emph{vectorization} of $A \in \R^{m\times n}$ -- denoted by $\vecto(A)$ -- is the vector formed by stacking the columns of $A$ into an $nm \times 1$ vector. For any matrices $A \in \R^{m\times n}$ and $B \in \R^{n\times m}$, we have that
\begin{enumerate}[label=\textup{(V\arabic*)}]
\item $\tr(AB) = \trasp{(\vecto(\trasp{A}))}\vecto(B)$; \label{prop:vec2}
\item $\vecto(AB) = (I_{m} \otimes A)\vecto(B)=(\trasp{B }\otimes I_{n})\vecto(A)$. \label{prop:vec3}
\end{enumerate}
\subsection{Convex Optimization and Contracting Dynamics}
\label{sec:contraction_theory}
We recall basic notions and results of convex optimization and contraction theory. 

\bd[Convex and strongly convex functions]
Let $\map{f}{\R^n}{\R}$ be a scalar function defined over a convex set $C \subseteq \R^n$. The function $f$ is (i) \emph{convex} if 
$$
f(\alpha x_1 + (1-\alpha) x_2) \leq \alpha f(x_1) + (1-\alpha) f(x_2), \quad \forall x_1 \neq x_2 \in C, \alpha \in {[0, 1]};
$$
(ii) \emph{strongly convex with parameter $m>0$} if 
$$
f(\alpha x_1 + (1-\alpha) x_2) \leq \alpha f(x_1) + (1-\alpha) f(x_2) + \frac{1}{2} m\alpha(1-\alpha)\norm{x-y}_2^2, \quad \forall x_1 \neq x_2 \in C, \alpha \in {[0, 1]}.
$$
\ed
Given a convex and differentiable function $\map{f}{\R^n}{\R}$, consider the \emph{optimization problem}
\beq\label{eq:unconstrained}
\min_{x \in \R^n} f(x).
\eeq

A paradigm that is becoming increasingly popular to solve OPs is to synthesize continuous-time dynamical systems that converge to an equilibrium, which corresponds to the optimal solution of the OP. In particular, a common approach to solve problem~\eqref{eq:unconstrained} is through the \emph{continuous-time gradient flow dynamics}, given by
\beq
\label{eq:grad_flow}
\dot x = - \nabla f(x).
\eeq
Significant efforts (see, e.g.,~\cite{RIK:60, WMH-SS:74, SL:84, PMW-JJES:20, ACBdO-MS-EDS:24}) have been directed toward characterizing the stability and convergence rates of the dynamics~\eqref{eq:grad_flow}, along with their robustness against uncertainty. In this context, a suitable tool to assess convergence is \emph{contraction theory}~\cite{WL-JJES:98, FB:24-CTDS}.
In contrast to standard convergence methods, rather than focusing on specific attractors, contraction theory focuses on characterizing how the distance between trajectories evolve. Intuitively, a system is contracting if any two trajectories converge (exponentially and without overshoot) towards each other.

Formally, consider a dynamical system $\dot{x}(t) = f\bigl(t,x(t)\bigr)$, where $\map{f}{\R_{\geq 0} \times C}{\R^n}$ is a smooth nonlinear function with $C\subseteq \R^n$ forward invariant set for the dynamics. The system is \emph{strongly contracting} if the distance, computed with respect to some norm $\norm{\cdot}$, between any two trajectories $x(\cdot)$ and $y(\cdot)$ rooted from initial conditions $x(0)$ and $y(0)$ diminishes exponentially with rate $c$, i.e.,
$$\|x(t) - y(t)\| \leq \e^{-ct}\|x(0) -y(0)\|, \quad \textup{ for all } t \geq 0.$$
If a system is contracting, then it exhibits ordered transient and asymptotic behaviors exhibited by contracting dynamics~\cite{WL-JJES:98, FB:24-CTDS, GR-MDB-EDS:10a, SX-GR-RHM:21}. In particular, initial conditions are exponentially forgotten, time-invariant dynamics admit a unique globally exponential stable equilibrium, and enjoy highly robust behaviors.

The following lemma, also known as Kachurovskii’s Theorem~\cite{RIK:60}, gives conditions under which the continuous-time gradient-flow dynamics~\eqref{eq:grad_flow} is strongly contracting.
\begin{lem}[Contractivity of continuous-time gradient flow dynamics]
\label{lem:gradient_flow}
Let $\map{f}{\R^n}{\R}$ be continuously differentiable. The following statements are equivalent:
\begin{enumerate}
\item $f$ is strongly convex with parameter $\nu>0$ (and problem~\eqref{eq:unconstrained} has a global minimum $x^\star \in \R^n$),
\item the gradient-flow dynamics~\eqref{eq:grad_flow} are strongly contracting with rate $\nu$ with respect to the $\ell_2$-norm (and has a unique equilibrium point $x^\star \in \R^n$).
\end{enumerate}
\end{lem}

\section{Similarity Matching}
\label{sec:similarity_matching}
Following the set-up from~\cite{CP-AMS-DBC:17}, consider $T$ centered input data samples\footnote{Data samples are \emph{centered} when adjusted to have zero mean.} in $\R^n$, denoted by $x_k \in \R^n$, $k \in \until{T}$. Let $X \in \R^{n \times T}$ be the input matrix obtained by concatenating these samples, i.e., $X := [x_1, \dots, x_T]$. The corresponding output representations in $\R^m$, $m \ll n$, are column-vectors $y_k \in \R^m$, concatenated into the matrix $Y := [y_1, \dots, y_T] \in \R^{m \times T}$.
The goal of the similarity matching problem is to find an output representation such that the similarity between each pair of output vectors matches the similarity between the corresponding input vectors. This \emph{similarity matching problem} can be formalized as:
\beq
\label{eq:similarity matching_matrix}
\min_{Y\in \R^{m \times T}} \frac{1}{T^2} \normF{\trasp{X}X - \trasp{Y}Y}^2 := \operatorname{SM(Y)},
\eeq
where: (i) the products $X^\top X$ and $Y^\top Y$ measure pairwise similarities between input and output samples, respectively; (ii) the cost is also known as similarity matching objective.
\subsection{Solution of the Similarity Matching Problem}
In \cite[Proposition 2.4]{XL-ZW-YZ:15} it is shown that the optimal solution of the OP~\eqref{eq:similarity matching_matrix} is given by the projections of input data into the principal subspace of their covariance matrix. Moreover, the problem has no local minima other than the principal subspace solution. First, we recall the following result, which establishes the equivalence between solving the OP~\eqref{eq:similarity matching_matrix} and computing
the dominant eigenspace of $\trasp{X}X$ associated with its non-negative eigenvalues and gives the explicit form of the minima of the OP~\eqref{eq:similarity matching_matrix}.
\begin{lem}
\label{lem:exact_sol_sm}
Assume that the matrix $\trasp{X}X \in \R^{T \times T}$ has $k>m$ positive eigenvalues. Let $\lambda_1 \geq \lambda_2 \geq \dots \geq \lambda_T \geq 0$ be the eigenvalues of $\trasp{X}X \in \R^{T \times T}$, and let $v_1, \dots, v_T \in \R^T$ be their corresponding unit eigenvectors.
Then, the point $Y^\star \in \R^{m \times T}$ is a minimizer of the OP~\eqref{eq:similarity matching_matrix} if and only if
$
Y^\star = U \Sigma^m \trasp{V_m},
$
where $V_m = (v_1, \dots, v_m)\in \R^{T\times m}$, $\Sigma^m = \diag{\sqrt{\lambda_1}, \dots ,\sqrt{\lambda_m}} \in \R^{m\times m}$, and $U \in \R^{m \times m}$ is an arbitrary orthogonal matrix.
\end{lem}

\subsection{Biologically Plausible Model}
\label{sec:biologically_plausible}
Starting from problem~\eqref{eq:similarity matching_matrix}, in~\cite{CP-AMS-DBC:17} a single-layer discrete time Hebbian/anti-Hebbian neural network (also termed as algorithm in~\cite{CP-AMS-DBC:17}) for dimensionality reduction is derived. The algorithm arises from a nested optimization problem obtained by embedding the OP~\eqref{eq:similarity matching_matrix} into a higher-dimensional space. Since this formulation is at the basis of the continuous-time network we consider and its convergence analysis (Section~\ref{sec:no_math}), we next provide a self-contained derivation of the embedding.  Consequently, we then recall the biologically plausible offline algorithm from~\cite{CP-AMS-DBC:17}.

\paragraph{Embedding problem~\eqref{eq:similarity matching_matrix} into a higher-dimensional space.}
Consider the OP~\eqref{eq:similarity matching_matrix}. By expanding the square of the Frobenius norm, the OP~\eqref{eq:similarity matching_matrix} is equivalent to
\beq
\label{eq:sm_vec}
\min_{Y\in \R^{m \times T}} \frac{1}{T^2} \tr(- 2\trasp{X}X\trasp{Y}Y + \trasp{Y}Y\trasp{Y}Y).
\eeq
Starting from the minimization problem~\eqref{eq:sm_vec},~\cite{CP-AMS-DBC:17} derived the following multi-variable OP which has the same $Y$-minima as problem~\eqref{eq:sm_vec} (we refer to Appendix~\ref{apx:eqW_M} for a self-contained derivation)
\beq
\label{eq:sm_offline}
\min_{W \in \R^{m\times n}} \max_{M \in \R^{m\times m}} \min_{Y\in \R^{m \times T}} 2\normF{W}^2 - \normF{M}^2 + \frac{2}{T}\wnormF{Y}{M}^2 - \frac{4}{T}\inprodF{WX, Y}.
\eeq

\paragraph{Biologically plausible offline optimization algorithm proposed by~\cite{CP-AMS-DBC:17}.}
Initialize the matrices $W \in \R^{m\times n}$ and $M\in \R^{m\times m}$, ensuring that $M$ is positive definite. The algorithm consists of two iterative steps, which are repeated until convergence:
\begin{enumerate}[label=(\arabic*)]
\item \emph{Optimization with respect to $Y$:} for fixed $W$ and $M$, solve the OP~\eqref{eq:sm_offline} with respect to $Y$, yielding:
\[
Y = M^{-1}WX := Y^\star.
\]
\item 
\label{item:learning_rules}
\emph{Optimization with respect to $M$ and $W$:} for fixed $Y$, perform a gradient descent-ascent step with respect to $W$ and $M$:
\[
\begin{aligned}
W &\leftarrow W + 2\eta \left(\frac{1}{T}Y\trasp{X} - W \right), \\
M &\leftarrow M + \frac{\eta}{\tau} \left(\frac{1}{T}Y\trasp{Y} - M\right),
\end{aligned}
\]
where $\tau >0$, and the step size $\eta \in {]0,1[}$, may depend on the iteration.
\end{enumerate}
\begin{rem}
    In the algorithm, $Y$ denotes the neuronal activity, while $W$ and $M$ represent the weights of the feedforward and lateral connections, respectively. The learning rules in~\ref{item:learning_rules} are local, following Hebbian and anti-Hebbian rules, respectively.

\end{rem}
\begin{rem}
\label{rem:positive_M}
In~\cite{CP-AMS-DBC:17} and related works it is stated (but not formally proven) that $M$ stays positive definite if initialized as
such. This property is what ensures the existence of the optimal $Y^\star$.
\end{rem}

\section{The Similarity Matching Network over Multiple Time Scales}
\label{sec:no_math}
We introduce the continuous-time neural network over multiple time scales that solves the similarity matching problem. This network 
features three interacting dynamics, co-evolving on different time scales. We describe the working principle of the network and then analyze the simpler scalar case to build intuition for the higher-dimensional framework. Given input data $X \in \R^{n \times T}$, consider the OP~\eqref{eq:sm_offline} with weight matrix $M \in \S_{\succ 0}^m$, that is the following \emph{embedded similarity matching problem}
\beq
\label{eq:sm_offline_general}
\min_{W \in \R^{m\times n}} \max_{M \in \S_{\succ 0}^m} \min_{Y\in \R^{m \times T}} 2\normF{W}^2 - \normF{M}^2 + \frac{2}{T}\wnormF{Y}{M}^2 - \frac{4}{T}\inprodF{WX, Y},
\eeq
and let $\map{S}{\R^{m\times n} \times \S_{\succ 0}^m \times \R^{m \times T}}{\R}$ denote its cost function. 
Adopting a time scale separation between neural dynamics, lateral synaptic dynamics, and feedforward synaptic dynamics, to solve the OP~\eqref{eq:sm_offline_general} we consider the following continuous-time dynamics, the \emph{similarity matching network over three time scales} (which we sometimes abbreviate as similarity matching network):
\begin{subequations}
\label{eq:3_time_scales_system}
\begin{empheq}[left=\empheqlbrace]{align}
\eps_1 \eps_2 \dot Y & = \frac{4}{T} \bigl(WX - MY \bigr), \quad \quad Y(0) := Y_0,\label{eq:fast}\\
\eps_2 \dot M & = - 2M + \frac{2}{T} Y Y^\top, \quad \quad M(0) := M_0, \label{eq:slow}\\
\dot W & = -4W + \frac{4}{T}Y X^\top, \quad \quad W(0) := W_0, \label{eq:super_slow}
\end{empheq}
\end{subequations}
where $0 < \eps_1 \eps_2 \ll \eps_2 \ll 1$ are the time-scale inducing parameters, $Y_0 \in \R^{m \times T}$, $M_0 \in \mathbf{S}^{m}_{>0}$, and $W_0 \in \R^{m \times n}$ are the initial conditions. 
Note that the feedforward synaptic dynamics~\eqref{eq:super_slow} follows a Hebbian rule, while the lateral synaptic dynamics~\eqref{eq:slow} follow an anti-Hebbian learning rule (due to the minus sign in the neural dynamics~\eqref{eq:fast}). The dynamics~\eqref{eq:3_time_scales_system} are schematically illustrated in Figure~\ref{fig:network}. In the figure: (i) to streamline presentation, panel a) shows the block scheme of the network for a single data sample; (ii) panel b) illustrates the block scheme for $T$ data samples.

\begin{figure}[h!]
\centering
\includegraphics[width=.9\linewidth]{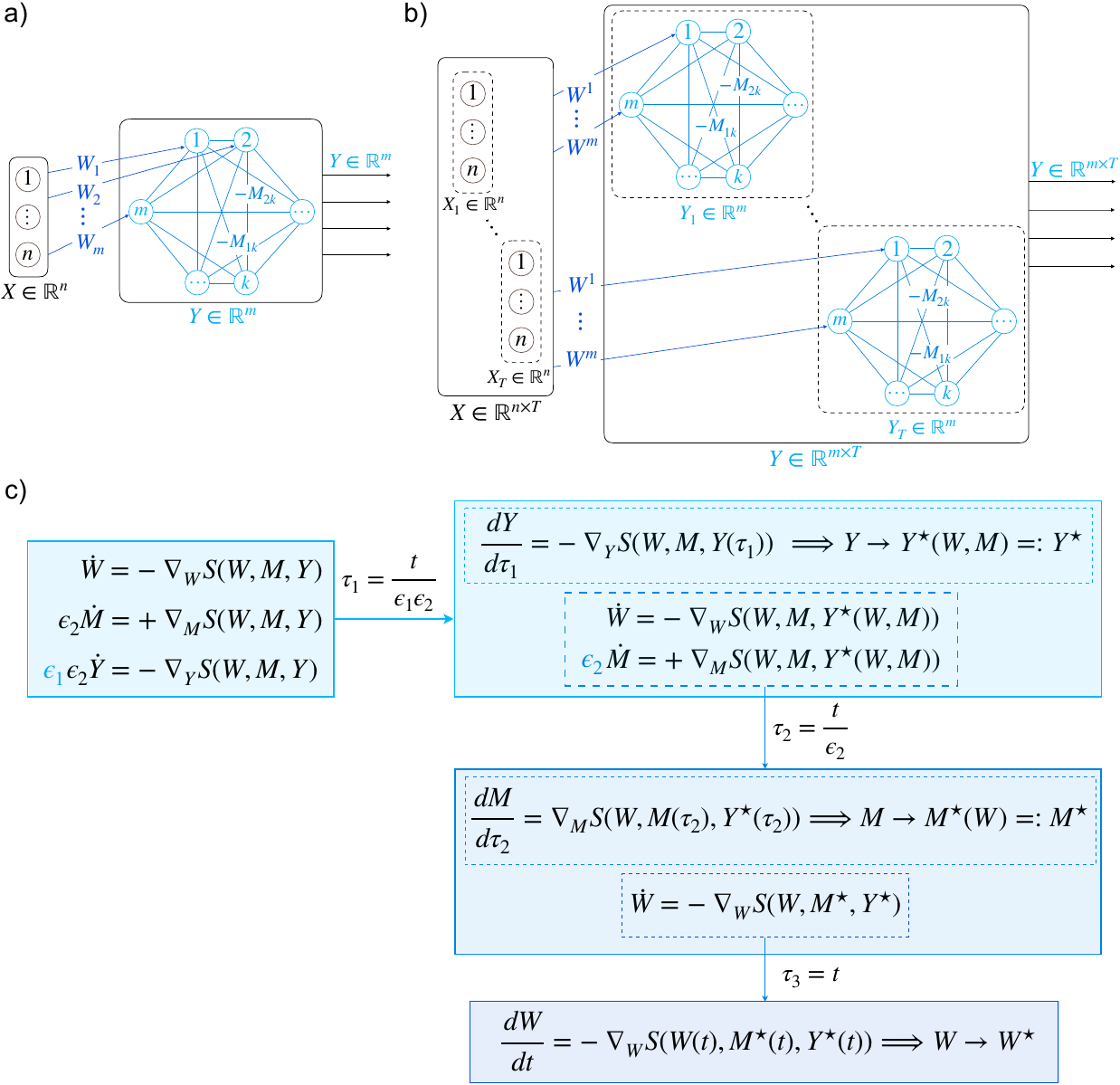}
\caption{
{Similarity matching network over three time scales. a) Single‐sample block scheme: the feedforward weights $W\in\R^{m\times n}$ project the input $X\in\R^n$ onto the hidden nodes $Y\in\R^m$, which evolve under the dynamics~\eqref{eq:fast} with lateral coupling $M\in\S^m_{>0}$. b) Block scheme of the considered network for $T$ data samples: batch extension the architecture described in panel a). c) Working principle of the network~\eqref{eq:3_time_scales_system}: on the fast scale $\tau_1$, the fast dynamics converges to the unique $Y^\star(W,M)$. Then, on the intermediate scale $\tau_2$, the intermediate dynamics converges to the unique $M^\star(W)$. Finally, on the slowest scale $\tau_3$, the slow dynamics converges to $W^\star$.}
}
\label{fig:network}
\end{figure}
Note that $\eps_1$ and $\eps_2$ induce time-scale separations between the three subsystems.
Specifically, equation~\eqref{eq:fast} is the \emph{fast subsystem}, equation~\eqref{eq:slow} is the \emph{intermediate subsystem}, and equation~\eqref{eq:super_slow} is the \emph{slow subsystem}.
Panel c) of Figure~\ref{fig:network} summarizes the {\em working principle} of the network. First, $\eps_1 \eps_2$ induces a time scale separation between the fast dynamics~\eqref{eq:fast} and the slower (synaptic) subsystems~\eqref{eq:slow} and~\eqref{eq:super_slow}. Essentially, considering the fast time scale $\tau_1 = \frac{t}{\eps_1\eps_2}$, the dynamics~\eqref{eq:fast} become
\beq
\label{eq:fast_time_1}
\frac{d Y}{d\tau_1} = - \nabla_Y S(W,M,Y(\tau_1)),
\eeq
where $W$ and $M$ are treated as fixed parameters. As $\epsilon_1 \to 0$, system~\eqref{eq:fast_time_1} converges to its equilibrium $Y^\star(W,M)$ (see Theorem~\ref{lem:layer1} for the formal result). The reduced slower subsystem is therefore
\begin{subequations}
\label{eq:3_time_scales_system_slower1}
\begin{empheq}[left=\empheqlbrace]{align}
\dot W & = - \nabla_W S(W,M,Y^\star(W,M)) \label{eq:super_slow_slow} \\
\eps_2 \dot M & = + \nabla_M S(W,M,Y^\star(W,M)).
\label{eq:slow_fast}
\end{empheq}
\end{subequations}
Again, for the dynamics in system~\eqref{eq:3_time_scales_system_slower1}  $\eps_2$ induces a time scale separation. This time, in the time scale $\tau_2 = \frac{t}{\eps_2}$, the dynamics~\eqref{eq:slow_fast} become
\beq
\label{eq:fast_time_2}
\frac{d M}{d\tau_2} = - \nabla_Y S(W,M(\tau_2),Y^\star(W,M(\tau_2))),
\eeq
where $W$ is treated as a fixed parameter. When $\epsilon_2 \to 0$, then system~\eqref{eq:fast_time_2} converges to $M^\star(W)$ (Theorem~\ref{lem:layer2} gives the formal result) and the slow system becomes
\beq
\label{eq:fast_time_3}
\dot W = - \nabla_W S(W,M^\star,Y^\star(W,M^\star)).
\eeq
Finally, system~\eqref{eq:fast_time_3} converges to $W^\star$ almost surely with random initialization (Theorem~\ref{lem:layer3}).
This time scale separation reflects the embedded similarity matching problem~\eqref{eq:sm_offline_general}, which can be written as the following three-level OP
\beq
\label{eq:multilevel_formulation}
\begin{aligned}
\min_{W \in \R^{m\times n}} & 
S(W, M^\star, Y^\star) \text{\quad where }\\
&M^\star = \argmax_{M \in \S_{\succ 0}^m} S(W, M, Y^\star) \text{\quad where }\\
& \quad \quad \quad Y^\star =\argmin_{Y\in \R^{m \times T}} S(W, M, Y),
\end{aligned}
\eeq
where $Y\in \R^{m \times T}$, $M \in \S_{\succ 0}^m$ and $W \in \R^{m\times n}$ are the variables at the first, second, and third level, respectively. For our convergence analysis, we focus on the multilevel OP in equation~\eqref{eq:multilevel_formulation}. Before formally diving into the analysis in Section~\ref{sec:convergence}, we first illustrate the simpler scalar case to build intuition for our technical framework. 
\subsection{Illustrating the Core Ideas Through the Scalar Case}
To better illustrate our approach, we first analyze the simpler scalar case where $n = m = T = 1$. In this case, the embedded similarity matching problem~\eqref{eq:sm_offline_general} simplifies to:
\beq
\label{eq:sm_offline_scalar}
\min_{w \in \R} \max_{m \in \R_{>0}} \min_{y\in \R} \ 2w^2 - m^2 + 2 my^2 - 4xwy.
\eeq
We let $S(w, m, y) := 2w^2 - m^2 + 2 my^2 - 4xwy$ be the cost function of the above problem.
In this scenario, $X^\top X = x^2 := \lambda_1$. Therefore, according to Lemma~\ref{lem:exact_sol_sm}, the solution to the similarity matching problem~\eqref{eq:similarity matching_matrix} is
\beq
\label{eq:exact_sol_scalar}
\subscr{y^\star}{sm} = \pm (x^2)^{1/2} = \pm x.
\eeq
In this case, the similarity matching network~\eqref{eq:3_time_scales_system} becomes
\begin{subequations}
\label{eq:3_time_scales_system_scalar}
\begin{empheq}[left=\empheqlbrace]{align}
\eps_1 \eps_2 \dot y & = 4 \bigl( wx - my\bigr), \quad \quad y(0) = y_0,\label{eq:fast_scalar}\\
\eps_2 \dot m & = - 2m + 2 y^2, \quad \quad m(0) = m_0>0, \label{eq:slow_scalar}\\
\dot w & = -4w + 4yx, \quad \quad w(0) = w_0. \label{eq:super_slow_scalar}
\end{empheq}
\end{subequations}
Next, we show that: (i) $\subscr{y^\star}{sm}$ is also obtained by solving problem~\eqref{eq:sm_offline_scalar} using a multilevel optimization approach; (ii) the network~\eqref{eq:3_time_scales_system_scalar} converges to {$(\subscr{y^\star}{sm}, m^\star, w^\star)$}.

\paragraph{First level: optimization with respect to $\mathbf{y}$.}
For any fixed $w, m \in \R$, with $m >0$, the map $y \mapsto S(w, m, y)$ is strongly convex and problem $\ds \min_{y\in \R} S(w, m, y)$ has a unique solution.
By minimizing problem~\eqref{eq:sm_offline_scalar} with respect to $y \in \R$, for each $w, m \in \R$, $m >0$, we obtain the optimal point
$
y^\star = m^{-1}w x.
$
For the network the fast dynamics~\eqref{eq:fast_scalar}, converges to the equilibrium point $y^\star$. The reduced slower subsystem therefore becomes
\begin{subequations}
\label{eq:3_time_scales_system_slower1_scalar}
\begin{empheq}[left=\empheqlbrace]{align}
\eps_2 \dot m & = - 2m + 2 m^{-2}w^2 x^2 \label{eq:slow_fast_scalar}\\
\dot w & = - 4w + 4m^{-1}w x^2 \label{eq:super_slow_slow_scalar}.
\end{empheq}
\end{subequations}

\paragraph{Second level: optimization with respect to $\mathbf{m}$.}
Next, for fixed $w \in \R$, consider the function
$$m \mapsto S_2(m) =: S(w, m, y^\star) = 2w^2-m^2 - 2 (w^2 x^2)m^{-1}.$$
This function is concave and the set $\R_{>0}$ is convex. Therefore, problem $\ds \max_{m\in \R_{>0}} S(w, m, y^\star)$ admits a unique solution. The function $S(w, m, y^\star)$ is illustrated in Figure~\ref{fig_m_w_scalar} for specific values of $w$ and $x$.
Maximizing $S(w, m, y^\star)$ with respect to $m >0$, for each $w \in \R$, we obtain the optimal point
$
m^\star = \bigl(w^2 x^2 \bigr)^{1/3}.
$
For the network dynamics, at the intermediate time scale~\eqref{eq:slow_fast_scalar} converges to the equilibrium point $m^\star$. The slow system then becomes
\beq
\label{eq:fast_time_3_scalar}
\dot w = - 4w + 4\bigl(w^2 x^2 \bigr)^{-1/3} w x^2 = - 4w + 4w^{1/3}x^{4/3}.
\eeq

\paragraph{Third level: optimization with respect to $\mathbf{w}$.}
We minimize the function 
$$w \mapsto S_3(w) =: S(w, m^\star, y^\star) = 2w^2 - 3 x^{4/3}w^{4/3}.$$
This function is non-convex and has two global minima, as illustrated in Figure~\ref{fig_m_w_scalar} for a given value of $x$ (it is straightforward to verify that this property holds for any value of $x$, we omit the computation here for the sake of space).
Minimizing $S_3(w)$ with respect to $w$ we obtain the optimal points
$
w^\star = \pm x^2 = \pm \lambda_1.
$
\begin{figure}[!h]
\centering
\includegraphics[width=0.495\linewidth]{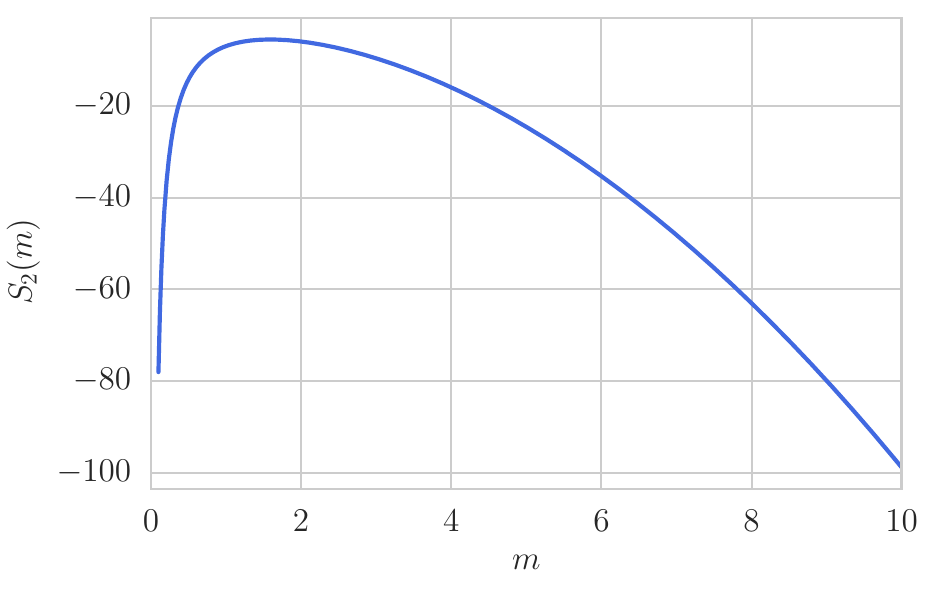}
\includegraphics[width=0.495\linewidth]{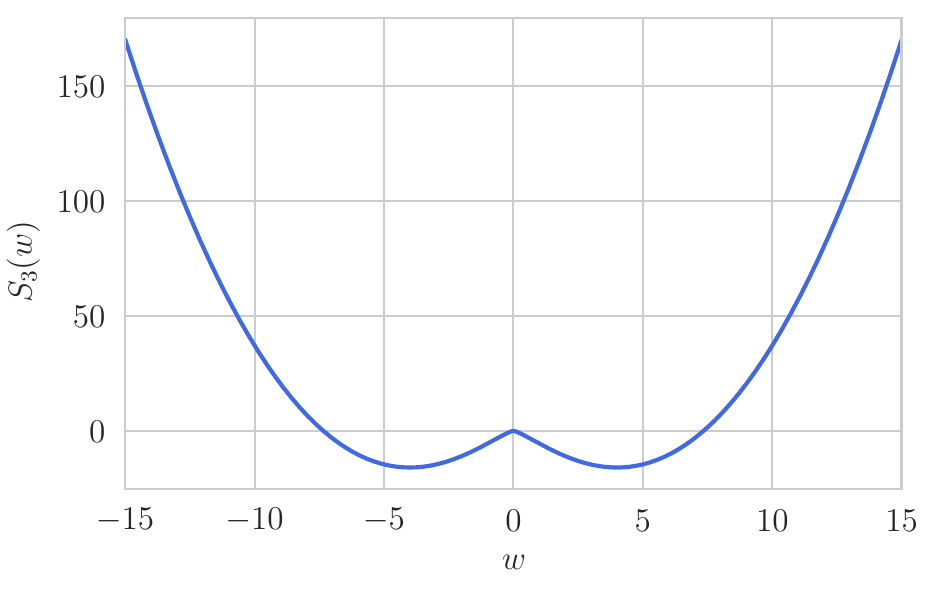}
\caption{Left Panel: Plot of the scalar function $S_2(m) = 2w^2 - m^2 - 2(w^2 x^2)m^{-1}$ for $w = 1$ and $x = 2$.
Right Panel: Plot of the scalar function $S_3(w) = 2w^2 - 3 x^{4/3}w^{4/3}$ for $x = 2$.}
\label{fig_m_w_scalar}
\end{figure}
{
The corresponding dynamics at the slow timescale~\eqref{eq:fast_time_3_scalar} has two locally stable equilibria $w^\star = \pm x^2$ and an unstable equilibrium $w^\star = 0$.
Figure~\ref{fig:scalar_dyn_w} shows the evolution of~\eqref{eq:fast_time_3_scalar} with $x=2$, starting from 1000 different initial conditions in the interval $[-15, 15]$. As shown in the figure, the dynamics~\eqref{eq:fast_time_3_scalar} converges to either $w^\star = x^2$ or $w^\star = -x^2$ for all initial conditions, except $0$.
}
\begin{figure}[!h]
\centering
\includegraphics[width=0.75\linewidth]{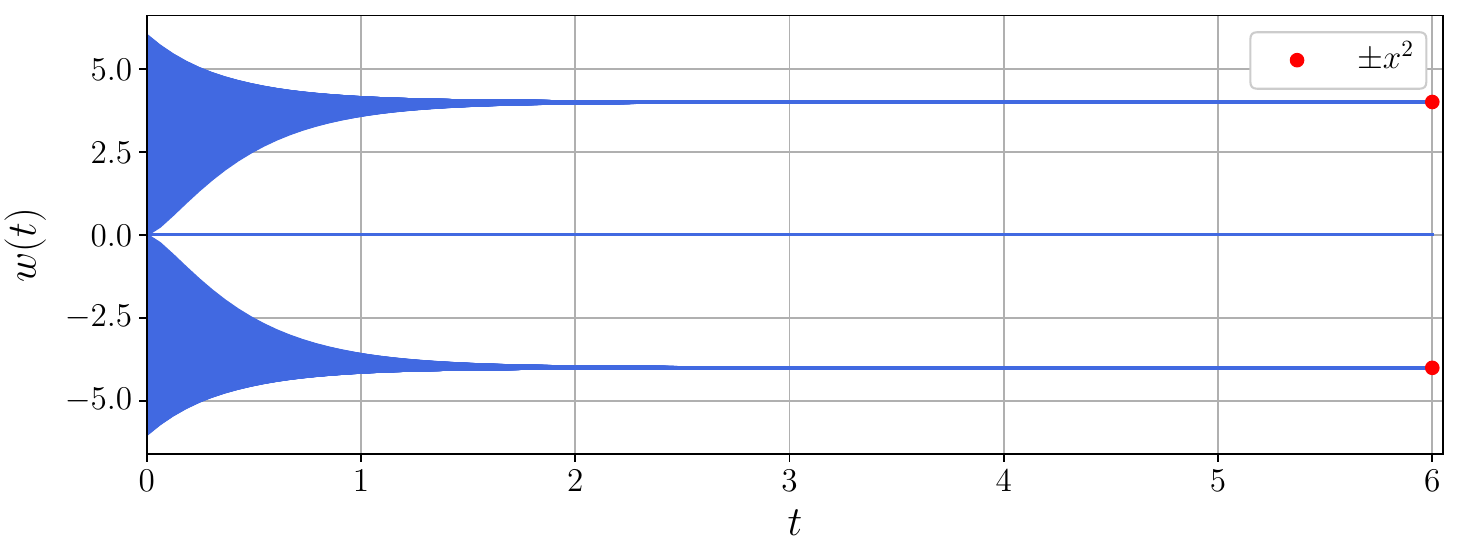}
\caption{Evolution of~\eqref{eq:fast_time_3_scalar} with $x=2$, starting from 1000 different initial conditions picked in the interval $[-15, 15]$. The dynamics converge either to $w^\star_2 = x^2$ or $w^\star_3 = -x^2$ for all initial conditions, except $0$.}
\label{fig:scalar_dyn_w}
\end{figure}

Finally, we project the optimal solutions back onto the output space to obtain the optimal $y^\star$. We compute \begin{align*}
y^\star &= (m^\star)^{-1}w^\star x = \bigl((w^\star)^2 \lambda_1\bigr)^{-1/3}\bigl( \pm \lambda_1\bigr) x =\bigl(\lambda_1^3\bigr)^{-1/3}\bigl( \pm \lambda_1\bigr) x = \pm x,
\end{align*}
which precisely matches the solution given in equation~\eqref{eq:exact_sol_scalar}.

As illustrated for this simpler case, the OPs of the first two levels benefit from having strongly convex (first level) or concave (second level) cost functions. In turn, this implies that the optimal solution in each level is unique and global and that the fast and intermediate dynamics converge to the optimal solution. Instead, the cost function of the last level is non-convex and this can yield multiple local minima. Indeed, in the scalar case considered above, the cost has two global minima and the corresponding dynamics converges to one of the minima almost for almost all initial conditions. Finally, when projecting back onto the output space, $y$ converges to the optimal  solution of the similarity matching problem.

In the multidimensional case, as we will see, we have a similar scenario. The first two levels feature strongly convex or concave cost functions but the third level's OP is non-convex, making the task of identifying the global optimum more intricate, but ultimately tractable. In summary, in the multidimensional case, we will obtain similar results as the scalar case.
\section{Main Results}
\label{sec:convergence}

We now present the main results of the paper. To streamline presentation, we first informally state the results. Then, we provide a formal convergence analysis for the similarity matching network~\eqref{eq:3_time_scales_system}. Within the analysis we characterize the behavior of each level in detail. Proofs for all results in this section are provided in Appendix~\ref{apx:proofs}.

Intuitively, the results presented can be informally summarized via the following:
\begin{metatheorem}[summary of the main results]
\label{metatheorem}
The trajectories of the similarity matching network over three time scales~\eqref{eq:3_time_scales_system} converge to an equilibrium point that is also a global optimal solution of the embedded similarity matching problem~\eqref{eq:sm_offline_general}. Specifically, the problem can be tackled with a multi-level optimization approach and:
\begin{enumerate}
\item the OP of the first (fast) level is a minimization problem with a strongly convex cost function and a unique minimizer, to which the neural dynamics globally exponentially converges;
\item the OP of the second (intermediate) level is a maximization problem with a strongly concave cost function and a unique maximizer, to which the lateral synaptic dynamics globally exponentially converges within the space of positive definite matrices. Additionally, the synaptic weights remain symmetric and positive definite for all time;
\item under two technical conjectures, the minimization problem of the third (slow) level admits multiple global minima, and the feedforward synaptic dynamics converges to one such minimizers from almost all initial conditions.
\end{enumerate}
Moreover, projecting the optimal solution back on the output space yields the optimal solution of the similarity matching problem~\eqref{eq:sm_vec}.
\end{metatheorem}
\begin{figure}[h!]
\centering
\includegraphics[width=0.9\linewidth]{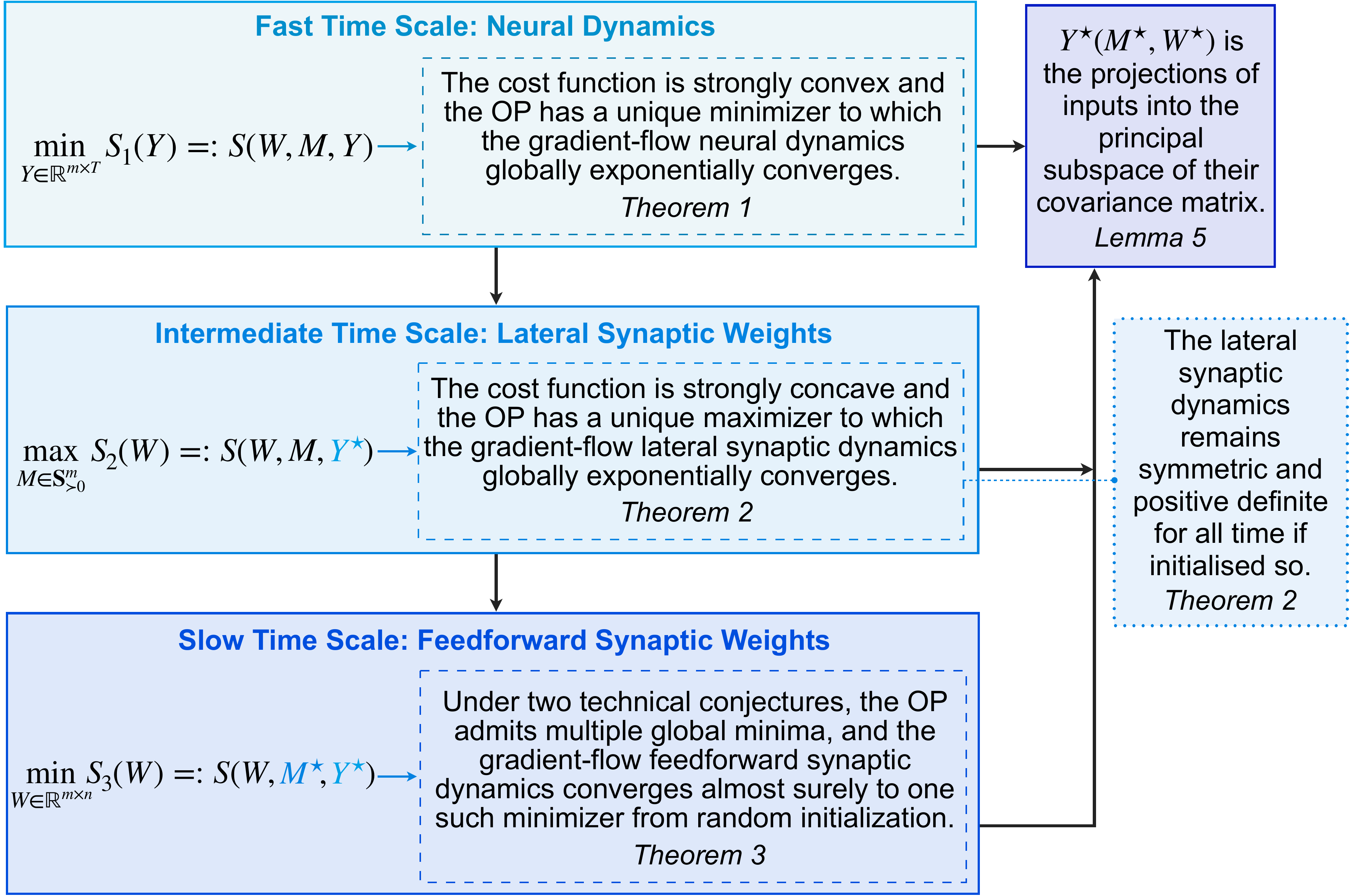}
\caption{Schematic diagram summarizing the main results in Section~\ref{sec:convergence}.}
\label{fig:schematic_results}
\end{figure}
{The assumptions, intermediate results, and their connections supporting the claims in the informal statement are summarized in Figure~\ref{fig:schematic_results}.} 
Specifically, we begin with the first level, corresponding to the fast time scale -- namely, that of the neural dynamics. For the associated minimization problem, we prove (see Theorem~\ref{lem:layer1}) that the cost function $S(W, M, Y)$ is strongly convex. This guarantees the existence of a unique optimal solution, say it $Y^\star = Y^\star(W, M)$, which we compute explicitly. We then derive the corresponding \emph{continuous-time gradient flow neural dynamics} and prove that these dynamics globally exponentially converge to $Y^\star$. We do so by showing strong contractivity of the dynamics.
Next, we analyze the second level, which operates on the intermediate time scale -- corresponding to the lateral synaptic weights. For this level's maximization problem, we show (see Theorem~\ref{lem:layer2}) that the cost function $S(W, M, Y^\star)$ is strongly concave in the space of positive definite matrices. This ensures a unique maximizer, denoted as $M^\star = M^\star(W)$, which we also compute explicitly. Moreover, we derive the corresponding \emph{continuous-time gradient flow lateral synaptic dynamics} and prove that these dynamics preserve symmetry and positive definiteness over time, if initialized so (see Lemma~\ref{item:3_second_layer}). Using contractivity arguments, we show that the dynamics converge exponentially to $M^\star$ within the space of positive definite matrices.
Finally, we consider the third level, associated with the slow time scale -- namely, that of the feedforward synaptic weights. Similar to the scalar case, the cost function at this level is non-convex, making the analysis more challenging. Nonetheless, we provide explicit expressions for the global minimizers, $W^\star$. We derive the corresponding \emph{continuous-time gradient-flow feedforward synaptic dynamics} and, under two technical conjectures (empirically validated in Appendix~\ref{apx:conjectures}), we prove almost sure convergence to $W^\star$ from random initialization.
Lastly, we project back the optimal solution onto the output space, i.e., we evaluate $Y^\star(W^\star, M^\star)$ (see Lemma~\ref{lem:final_proj}). As desired, this yields the optimal solution of the similarity matching problem~\eqref{eq:sm_vec}, as given in Lemma~\ref{lem:exact_sol_sm}.
\subsection{Fast Time Scale -- Neural Dynamics}
Given $M \in \S_{\succ 0}^m$, $W \in \R^{m\times n}$, and $X \in \R^{n \times T}$, the minimization problem of the first level, which corresponds to the fast time scale, is
\beq
\label{eq:op_layer1}
\min_{Y \in \R^{m \times T}}\frac{2}{T}\wnormF{Y}{M}^2 - \frac{4}{T}\inprodF{WX, Y} + 2\normF{W}^2 - \normF{M}^2.
\eeq
Let $\map{S_1}{\R^{m \times T}}{\R}$ denote the cost function of~\eqref{eq:op_layer1}. Next, we show that the map $S_1(Y)$ is strongly convex, we explicitly compute the minimizer of problem~\eqref{eq:op_layer1}, provide the gradient-flow dynamics, and show that this exponentially converges to the minimizer.

\bt[Analysis of the first level]
\label{lem:layer1}
Consider the minimization problem~\eqref{eq:op_layer1}. For any $M \in \S_{\succ 0}^m$, $W \in \R^{m\times n}$, and $X \in \R^{n \times T}$, we have
\begin{enumerate}
\item \label{item:1_first_layer}
the cost function $S_1$ is strongly convex with parameter ${\nu}_{Y} = \frac{4}{T}\subscr{\lambda}{min}(M)$;
\item \label{item:2_first_layer}
the global minimizer of the OP~\eqref{eq:op_layer1} is
\beq
\label{eq:ystar}
Y^\star = M^{-1}WX;
\eeq
\item
\label{item:3_first_layer}
the \emph{continuous-time gradient-flow neural dynamics} in the fast time scale
\beq
\label{eq:dyn_first_layer}
\dot Y = -\nabla_Y S(W, M, Y) = \ds - \frac{4}{T} \bigl( MY - WX \bigr)
\eeq
is strongly contracting with rate ${\nu}_{Y} = \frac{4}{T}\subscr{\lambda}{min}(M)$ with respect to the $\ell_2$-norm.
\end{enumerate}
\et
Statement~\ref{item:3_first_layer} implies that any solution of the neural dynamics~\eqref{eq:dyn_first_layer} globally exponentially converges towards the unique equilibrium point $Y^\star$. Moreover, the system exhibits all the robust properties of contracting systems (see Section \ref{sec:contraction_theory}).
\subsection{Intermediate Time Scale -- Lateral Synaptic Weights}
First, we determine the cost function for the second level, which corresponds to the intermediate time scale. We compute 
\begin{align*}
S(W, M, Y^{\star}) = 2\normF{W}^2 - \normF{M}^2 - \frac{2}{T}\tr(\trasp{W}M^{-1}WX\trasp{X})
= 2\normF{W}^2 - \normF{M}^2 - 2\inprodF{M^{-1}, WC_X\trasp{W}},
\end{align*}
where $\ds C_X := \frac{X\trasp{X}}{T}$ is the covariance matrix of the centered input data samples.
Therefore, for given $W \in \R^{m\times n}$ and $X \in \R^{n \times T}$, the OP of the second level is 
\beq
\label{eq:op_layer2}
\max_{M \in \S_{\succ 0}^m} 2\normF{W}^2 - \normF{M}^2 - 2\inprodF{M^{-1}, WC_X\trasp{W}}.
\eeq
Denote by $\map{S_2}{\S_{\succ 0}^m}{\R}$ the cost function of the OP~\eqref{eq:op_layer2}. Next, we show that the map $S_2$ is strongly concave, explicitly compute the maximizer of~\eqref{eq:op_layer2}, provide the gradient-flow dynamics for this level, and show that this exponentially converges in $\S_{\succ 0}^m$.
\bt[Analysis of the second level]
\label{lem:layer2}
Consider the maximization problem~\eqref{eq:op_layer2}.
For any full rank $W \in \R^{m\times n}$ and $X \in \R^{n \times T}$ such that $C_X \succ 0$, we have
\begin{enumerate}
\item \label{item:1_second_layer}
the cost function $S_2$ is strongly concave with parameter ${\nu}_{M} = 2$;
\item \label{item:2_second_layer}
the global maximizer of problem~\eqref{eq:op_layer2} is
\beq
\label{eq:Mstar}
M^\star = \bigl(WC_X\trasp{W}\bigr)^{1/3};
\eeq
\item \label{item:3_second_layer}
for any $\map{Y}{\R_{\geq0}}{\R^{m \times T}}$, the gradient-flow lateral synaptic dynamics 
\beq
\label{eq:dynamics_M}
\dot M(t) = - M(t) + \frac{1}{T} Y(t) \trasp{Y(t)}, \quad M_0\in \S_{\succ 0}^m,
\eeq
satisfy the following
\begin{enumerate}
\item $M(t) = \trasp{M}(t), \quad \forall t >0$;
\label{item1:lem_M}
\item $M(t) \succ 0, \quad \forall t >0$;
\label{item2:lem_M}
\end{enumerate}
\item \label{item:4_second_layer}
the \emph{continuous-time gradient-flow lateral synaptic dynamics} in the intermediate time scale
\beq
\label{eq:dyn_second_layer}
\dot M = \nabla_M S(W, M, Y^\star) = \ds - 2M + 2 M^{-1}WC_X\trasp{W}M^{-1},
\eeq
with $M_0\in\S_{\succ 0}^m$,
is strongly contracting with rate ${\nu}_{M} = 2$ with respect to the $\ell_2$-norm in $\S_{\succ 0}^m$.
\end{enumerate}
\et
Before presenting our result for the slow time scale, we make a number of comments on Theorem \ref{lem:layer2}.
First, statement~\ref{item:3_second_layer} provides a key invariance result for the analysis of problem~\eqref{eq:sm_offline_general}. While we show this property in continuous time, this invariance result holds as well in discrete-time. {As noted in Remark~\ref{rem:positive_M}, the matrix $M$ in the algorithm in Section~\ref{sec:biologically_plausible} remains positive definite if initialized so. However, to the best of our knowledge, item~\ref{item:3_second_layer} represents the first formal proof of this property.} With a slight abuse of terminology, we sometimes refer to this property as “positive in a matrix sense".
Statement~\ref{item:4_second_layer} implies that any solution of the gradient-flow lateral synaptic dynamics~\eqref{eq:dyn_second_layer} starting from initial conditions $M_0\in\S_{\succ 0}^m$ exponentially converge in $\S_{\succ 0}^m$ towards the equilibrium point $M^\star$. Moreover, the system exhibits all the robust properties characteristic of contracting systems. Finally, as a consequence of statement~\ref{item:3_second_layer}, the dynamics~\eqref{eq:dyn_second_layer} is positive in a matrix sense.

\subsection{Slow Time Scale -- Feedforward Synaptic Weights}
\label{sec:third_layer}
We now analyze the third level. To derive the corresponding cost function, we compute
\begin{align*}
S(W, M^{\star}, Y^{\star}) &= 2\normF{W}^2 - \normF{\bigl(WC_X\trasp{W}\bigr)^{1/3}}^2 - 2\inprodF{\bigl(WC_X\trasp{W}\bigr)^{-1/3}, WC_X\trasp{W}}\\
&= 2\normF{W}^2 - \normF{\bigl(WC_X\trasp{W}\bigr)^{1/3}}^2 - 2\normF{\bigl(WC_X\trasp{W}\bigr)^{1/3}}.
\end{align*}
Therefore, given $X \in \R^{n \times T}$ such that $C_X \succ 0$, the minimization problem in the third level is
\beq
\label{eq:op_layer3}
\min_{W \in \R^{m\times n}} 2\normF{W}^2 - 3\normF{\bigl(WC_X\trasp{W}\bigr)^{1/3}}^2.
\eeq
We let $\map{S_3}{\R^{m\times n}}{\R}$ denote the cost function of problem~\eqref{eq:op_layer3}. Note that $S_3(W)$ is non-convex, continuous, and not differentiable on $\R^{m\times n}$ due to the matrix fractional power. However, it is differentiable on the set of full rank matrices $\mcR_{m,n} := \setdef{W \in \R^{m\times n}}{\rank(W) = m}$.
With the next result, we show that the stationary points of~\eqref{eq:op_layer3} in $\mcR_{m,n}$ are the full rank matrices whose eigenvalues correspond to $m$ eigenvalues of $C_X$, with the associated right eigenvectors being the corresponding eigenvectors of $C_X$.
\begin{lem}[Stationary points of~\eqref{eq:op_layer3}]
\label{lem:stat_points}
Given $X \in \R^{n \times T}$ such that $C_X \succ 0$, consider the minimization problem~\eqref{eq:op_layer3} in $\mcR_{m,n}$.
The matrix $\bar{W} \in \mcR_{m,n}$ is a stationary point of the OP~\eqref{eq:op_layer3} if and only if 
$$\bar{W} \in \mathcal{S} := \bigsetdef{W \in \mcR_{m,n}}{W = U \left(P \Lambda_C P^\top \right)_{m,m} (PV_C)_m^\top,~P\in \mathcal{P}_{n},~U \in \mathcal{O}_m},$$
where $(V_C, \Lambda_C)$ is the SVD of $C_X$, $\left(P \Lambda_C P^\top \right)_{m,m}$ denotes the leading $m\times m$ principal submatrix of $P \Lambda_C P^\top$, and $(PV_C)_m$ is the submatrix of the first $m$ columns of $PV_C$.
\end{lem}
Another key property of the OP~\eqref{eq:op_layer3} is that its cost function is coercive. Namely, the function value tends to infinity as the norm of its input grows to infinity. We formalize this in the following lemma.
\begin{lem}
\label{lem:coercive}
Given $X \in \R^{n \times T}$ such that $C_X \succ 0$, consider the minimization problem~\eqref{eq:op_layer3} in $\mcR_{m,n}$. Then the cost function $S_3$ is coercive, that is $\lim_{\normF{W}^2 \to \infty}S_3(W) = + \infty$.
\end{lem}
The above property implies that any sublevel set of $S_3$ is bounded. This is key for the analysis of the gradient flow, as it ensures that trajectories cannot escape to infinity.
Following the analysis of the previous two levels, we are interested in the gradient-flow dynamics at the slow time scale, that is,
\beq
\label{eq:dyn_third_layer}
\dot W = - \nabla_W S(W, M^{\star}, Y^{\star}) = - 4W + 4\bigl(WC_X\trasp{W}\bigr)^{-1/3} WC_X, \quad W(0) = W_0.
\eeq
The RHS of~\eqref{eq:dyn_third_layer} is well defined only if $W(t)\in \mcR_{m,n}$, for all $t$. 
\begin{quote} \emph{Heuristic motivation for Conjecture~\ref{conj_full_row_rank}.}
Being $\dot W$ the gradient of the cost $S_3$ (and therefore, the cost function decreases along the flow trajectories), it is reasonable that trajectories initialized in $\mcR_{m,n}$ remain full rank. Indeed, as the flow approaches a rank-deficient matrix, the cost increases while its gradient goes to infinity. 
\end{quote}
This heuristic analysis motivates the following conjecture, that we empirically validate in Appendix~\ref{apx:conjecture1}.
\begin{conj}[Full-rank invariance]
\label{conj_full_row_rank}
Given $X \in \R^{n \times T}$ such that $C_X \succ 0$, consider the \emph{continuous-time gradient-flow feedforward synaptic dynamics}~\eqref{eq:dyn_third_layer}. If $W_0 \in \mcR_{m,n}$, then $W(t)\in \mcR_{m,n}$ for all $t > 0$.
\end{conj}
\begin{quote} \emph{Heuristic motivation for Conjecture~\ref{conj_full_strict_saddle}.}
Let $\bar{W} \in \mathcal{S}$ and compute
\begin{align*}
S_3(\bar{W}) &= 2\normF{\bar{W}^2} - 3\normF{\bigl(\bar{W} C_X \bar{W}^\top\bigr)^{1/3}}^2 \\
&= 2\left\|\left( P \Lambda_C P^\top \right)_{(m,m)}\right\|_{\textup{F}}^2 - 3\left\|\left(\left(P \Lambda_C P^\top \right)_{(m,m)} \left(P \Lambda_C P^\top \right)_{(m,m)} \left( P \Lambda_C P^\top \right)_{(m,m)}\right)^{1/3}\right\|_{\textup{F}}^2 \\
&= - \left\|\left( P \Lambda_C P^\top \right)_{(m,m)}\right\|_{\textup{F}}^2.
\end{align*}
The above norm is minimized when the leading $m \times m$ principal submatrix contains the $m$ largest eigenvalues of $\Lambda_C$. Therefore, the global minima of problem~\eqref{eq:op_layer3} are the points $W^\star = U \Lambda_C^m (V_C^m)^\top$, for any $U \in \mathcal{O}_m$ and $P \in \mathcal{P}_n$ where $\Lambda_C^m = \diag{\lambda_1^C, \dots, \lambda_m^C}$, and $V_C^m = (v^C_1, \dots, v^C_m)\in \R^{n\times m}$. 
\end{quote}
The above heuristic analysis motivates the following conjecture, that we empirically validate in Appendix~\ref{apx:conjecture2}.
\begin{conj}[Global minima and strict saddles]
\label{conj_full_strict_saddle}
Given $X \in \R^{n \times T}$ such that $C_X \succ 0$, consider the continuous-time gradient-flow feedforward synaptic dynamics~\eqref{eq:dyn_third_layer}. The global minima of $S_3(W)$ are the points $W^\star \in \mathcal{S}$ such that $W^\star = U \Lambda_C^m (V_C^m)^\top$, with $U \in \mathcal{O}_m$. Any other $\bar W \in \mathcal{S}$ is a strict saddle or a maximum.
\end{conj}
We can now give the following result on the convergence of the {gradient-flow dynamics}~\eqref{eq:dyn_third_layer}.
\bt[Analysis of the third level]
\label{lem:layer3}
Under Conjectures~\ref{conj_full_row_rank} and~\ref{conj_full_strict_saddle}, for almost all initial conditions $W_0 \in \mcR_{m,n}$ and for any input data $X \in \R^{n \times T}$ such that $C_X \succ 0$, the \emph{continuous-time gradient-flow feedforward synaptic dynamics} in the slow time scale~\eqref{eq:dyn_third_layer} converges to a global minimizer $W^\star = U \Lambda_C^m (V_C^m)^\top$, with $U \in \mathcal{O}_m$.
\et
\subsection{Optimal Final Solution}
Finally, we project the optimal solutions back into the output space to obtain the optimal $Y^\star$. The following result shows that, as expected, $Y^\star$ is equal to the optimal solution of the similarity matching OP~\eqref{eq:similarity matching_matrix} given in Lemma~\ref{lem:exact_sol_sm}.

\begin{lem}
\label{lem:final_proj}
Let $Y^\star(M, W)$, $M^\star(W)$, and $W^\star$ be the optimal solutions of the first, second, and third levels given in Theorems~\ref{lem:layer1}
,~\ref{lem:layer2}, and~\ref{lem:layer3}, respectively. Then, $Y^\star(M^\star, W^\star)$ is the projection of input data into the principal subspace of their covariance matrix.
\end{lem}
\section{Simulations}
\label{sec:simulations}
We illustrate the effectiveness of the similarity matching network~\eqref{eq:3_time_scales_system} in solving a principal subspace projection problem via a numerical example\footnote{The code to replicate all the simulations in this section is available at the GitHub \url{https://shorturl.at/l1qVf}.}.
This example, built upon the one in~\cite{CP-AMS-DBC:17}, serves three main purposes: (i) to validate the theoretical results presented in Section~\ref{sec:convergence}; (ii) to compare our method with the algorithm proposed in~\cite{CP-AMS-DBC:17}; and (iii) to investigate the influence of the time-scale parameters $\eps_1$ and $\eps_2$ in equation~\eqref{eq:3_time_scales_system}.

To this end, we consider an $n = 10$ dimensional dataset composed by $T = 2000$ centered input data samples and simulate the system~\eqref{eq:3_time_scales_system} to reduce this dataset to $m = 3$ dimensions.
For each simulation in this section, following~\cite{CP-AMS-DBC:17}, we generate the input data matrix $X \in \R^{10 \times 2000}$ from its SVD. Specifically, the top $m$ singular values are set to $\sqrt{3T}, \sqrt{2T}$, and $\sqrt{T}$ while the remaining singular values are uniformly sampled from the range $[0, 0.1\sqrt{T}]$. The left and right singular vectors are randomly generated. The initial conditions are set as follows: $Y_0 \in \R^{m \times T}$ and $W_0 \in \R^{m \times n}$ are initialized with entries drawn from a standard normal distribution, while $M_0 \in \R^{m \times m}$ is a diagonal matrix with positive entries sampled from a standard normal distribution. Simulations of the dynamics~\eqref{eq:3_time_scales_system} are performed with Python using the ODE solver \emph{solve\_ivp}.
Finally, we let $Y^\star$ denote the optimal solution of the similarity matching OP~\eqref{eq:similarity matching_matrix} given in Lemma~\ref{lem:exact_sol_sm}, and use the subscript “$\textup{P}$" and “$\textup{dyn}$" to denote the numerical solutions obtained via the algorithm in~\cite{CP-AMS-DBC:17} and our dynamics~\eqref{eq:3_time_scales_system}, respectively.

\paragraph{Comparison with the algorithm in~\cite{CP-AMS-DBC:17}.}
We run 500 simulations of both the similarity matching network over three time scales~\eqref{eq:3_time_scales_system} and the offline algorithm in~\cite{CP-AMS-DBC:17} (see Section~\ref{sec:biologically_plausible} for a self-contained review). Both methods are initialized with the same randomly generated initial conditions. 
While the full rank assumption is needed for the offline algorithm in~\cite{CP-AMS-DBC:17} to converge, and it is also theoretically needed in Theorem~\ref{lem:layer3} to ensure that the dynamics~\eqref{eq:dyn_third_layer} is well-defined, time-scale separation empirically guarantees convergence even when $W_0$ is not full rank. To verify this, we initialize $W_0$ as not non-full-rank matrix in three selected simulations. Specifically, in simulation $100$ and $200$, $W_0$ is initialized with rank equal to $2$ and $1$, respectively, while in simulation $300$, $W_0$ is set to the zero $m\times n$ matrix.
For the algorithm in~\cite{CP-AMS-DBC:17} we set $\tau = 0.5$, and $\eta = 0.01$, while for the coupled dynamics~\eqref{eq:3_time_scales_system} we set $\eps_1 = 0.01$ and $\eps_2 = 0.5$.
Figure~\ref{fig:comparison_cost_function} shows the cost difference $\operatorname{SM}(Y) - \operatorname{SM}(Y^\star)$ for both methods, where $Y$ is either $\subscr{Y}{P}$ or $\subscr{Y}{dyn}$. The results empirically show that, under this setup, the dynamics~\eqref{eq:3_time_scales_system} converge to the optimal solution $Y^\star$, achieving an error of the same order of magnitude across all simulations. In contrast, the offline algorithm in~\cite{CP-AMS-DBC:17} fails to converge when $W_0 \notin \mcR_{m,n}$.
Moreover, the simulations empirically confirm that each dynamic of the system~\eqref{eq:3_time_scales_system} converges to the respective theoretical optimal solutions, as established in Theorems~\ref{lem:layer1}--\ref{lem:layer3}. For brevity, we do not include these plots.
\begin{figure}[!h]
\centering 
\includegraphics[width=.8\linewidth]{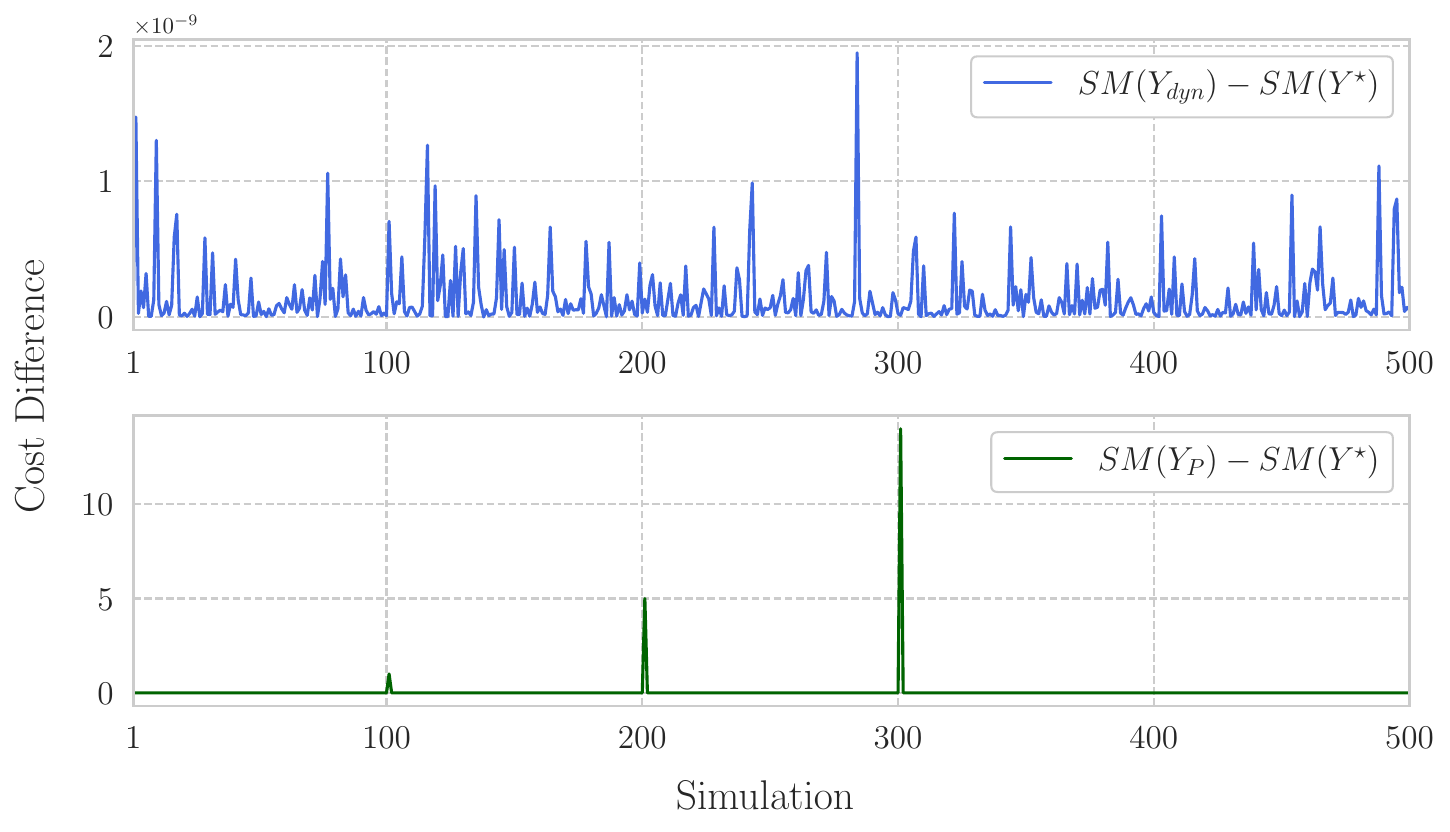}
\caption{Cost difference $\operatorname{SM}(Y) - \operatorname{SM}(Y^\star)$ for the proposed dynamics~\eqref{eq:3_time_scales_system} (up) and the algorithm in~\cite{CP-AMS-DBC:17} (bottom) over 500 simulations. Both methods start from the same randomly generated initial condition. In simulation $100$ and $200$, $W_0$ is initialized with rank equal to $2$ and $1$, respectively, while in simulation $300$, $W_0$ is set to the zero $m\times n$--matrix.
We empirically observe that the dynamics~\eqref{eq:3_time_scales_system} consistently converge to the optimal theoretical solution $Y^\star$ in each simulation. In contrast, the algorithm in~\cite{CP-AMS-DBC:17} fails in the simulations where $W_0$ is not full rank.}
\label{fig:comparison_cost_function}
\end{figure}

\paragraph{Robustness to noise.}
In Theorems~\ref{lem:layer1} and~\ref{lem:layer2}, we demonstrate that the dynamics of $Y$ and $M$ are strongly contracting. This property, in turn, implies several desirable properties, including robustness to noise.
To test this robustness, we run 10 simulations of system~\eqref{eq:3_time_scales_system}, introducing Gaussian noise into the dynamics of $Y$ and $M$. We set $\eps_1 = 0.01$, $\eps_2 = 0.2$ and plot the cost difference $\operatorname{SM}(\subscr{Y}{dyn}) - \operatorname{SM}(Y^\star)$ in Figure~\ref{fig:cost_noise}. Consistently with the contractivity results in Theorems~\ref{lem:layer1} and~\ref{lem:layer2}, the system~\eqref{eq:3_time_scales_system} converges to the theoretical optimal solution, even in the presence of noise.
\begin{figure}[!h]
\centering 
\includegraphics[width=.95\linewidth]{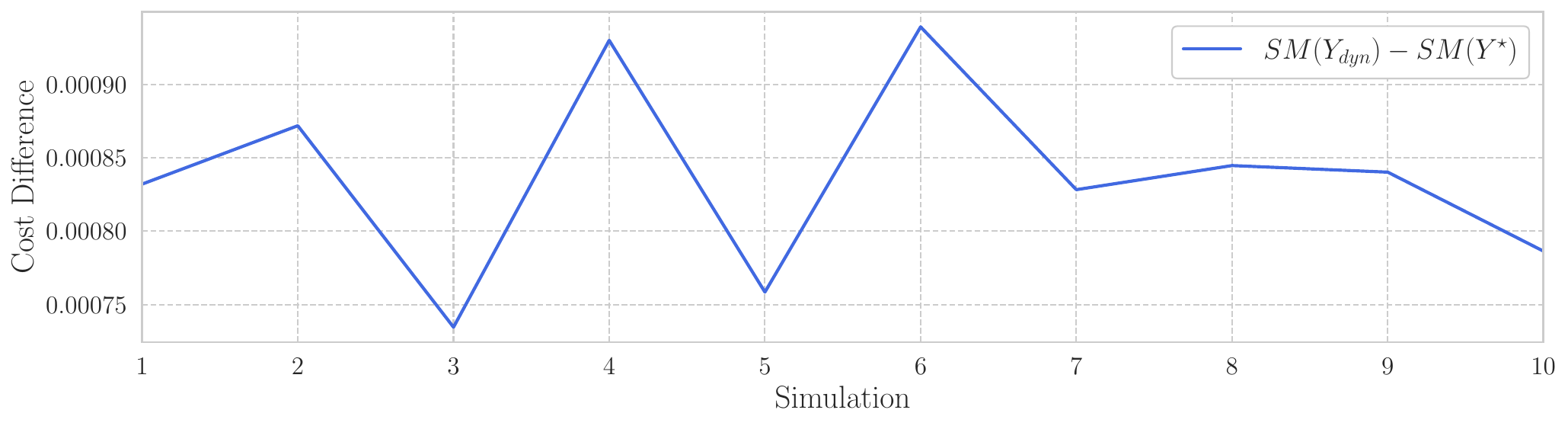}
\caption{Cost difference $\operatorname{SM}(\subscr{Y}{dyn}) - \operatorname{SM}(Y^\star)$ over 10 simulations of system~\eqref{eq:3_time_scales_system} with Gaussian noise added to the contracting dynamics of the first and second level. In line with the contractivity results in Theorems~\ref{lem:layer1} and~\ref{lem:layer2}, the system converges to the theoretical optimal solution.}
\label{fig:cost_noise}
\end{figure}

\paragraph{Different time-scale parameters.}
To investigate the effect of the time-scale parameters $\eps_1$ and $\eps_2$, we run simulations across ten distinct $(\eps_1, \eps_2)$ pairs: $(0.001, 0.05)$, $(0.005, 0.1)$, $(0.01, 0.1)$, $(0.02, 0.15)$, $(0.02, 0.2)$, $(0.03, 0.25)$, $(0.03, 0.3)$, $(0.05, 0.4)$, $(0.05, 0.5)$, and $(0.2, 0.5)$. For each pair, we run 100 simulations of the dynamics~\eqref{eq:3_time_scales_system}. The results empirically confirm that, for every $(\eps_1, \eps_2)$ pair, system~\eqref{eq:3_time_scales_system} consistently converges to the theoretical optimal solution $Y^\star$. For the sake of space, Figure~\ref{fig:comparison_cost_function} shows the cost difference $\operatorname{SM}(\subscr{Y}{dyn}) - \operatorname{SM}(Y^\star)$ for six selected $(\eps_1, \eps_2)$ pairs.
\begin{figure}[!h]
\centering 
\includegraphics[width=.95\linewidth]{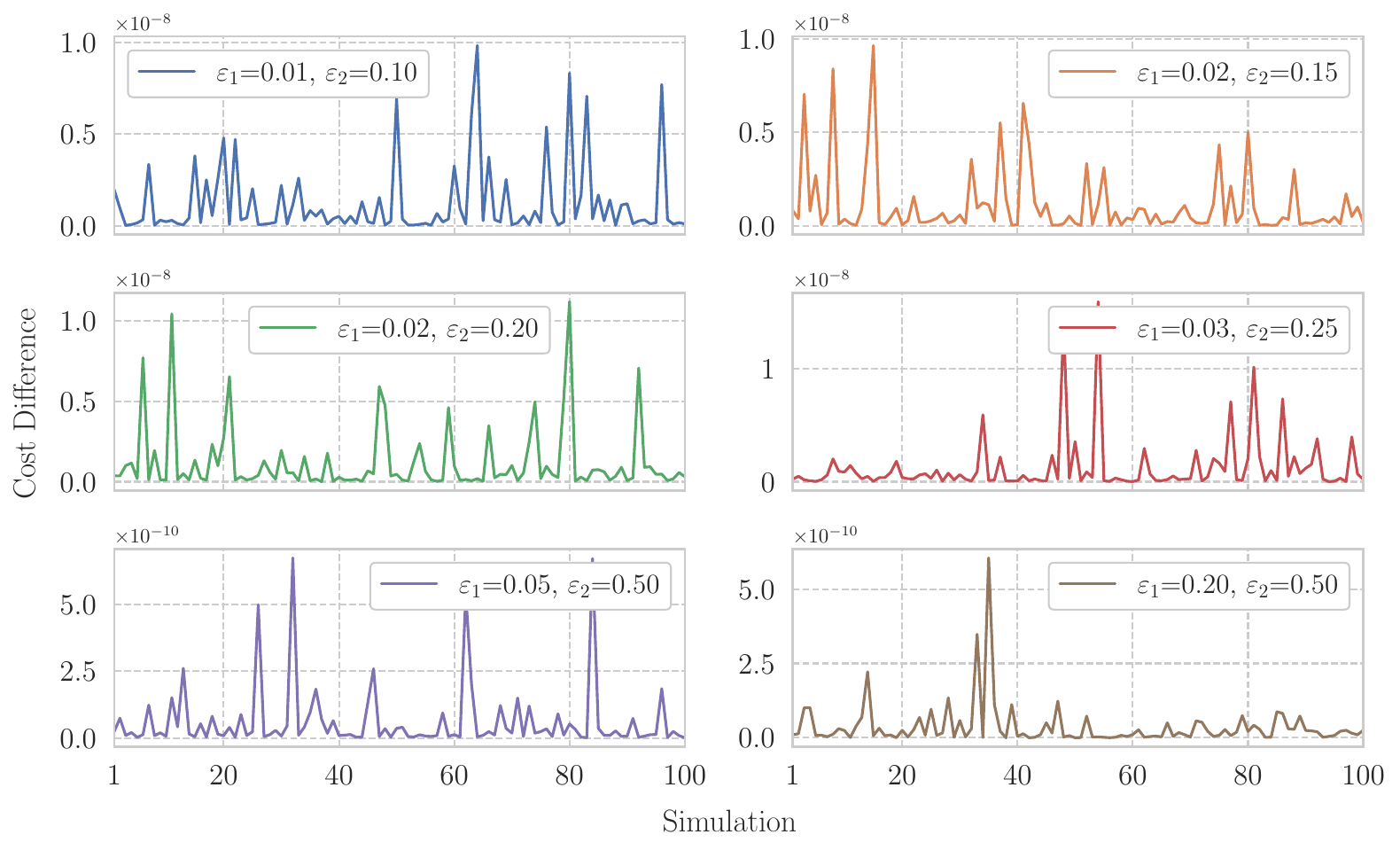}
\caption{Cost difference $\operatorname{SM}(\subscr{Y}{dyn}) - \operatorname{SM}(Y^\star)$ over 100 simulations of system~\eqref{eq:3_time_scales_system} for six different $(\eps_1, \eps_2)$ pairs. In each case, the dynamics~\eqref{eq:3_time_scales_system} converge to the theoretical optimal solution $Y^\star$.}
\label{fig:diff_eps}
\end{figure}
\section{Conclusion}
\label{sec:conclusions}
In this paper, we considered the similarity matching network over three time scales: a continuous-time biologically plausible neural network solving the similarity matching problem for principal subspace projection.
Building on the framework proposed in~\cite{TH-CP-DBC:14, CP-AMS-DBC:17}, our network is derived optimizing a min-max-min principled objective that embeds the similarity matching problems into a higher-dimensional space, that is the space of neural dynamics, lateral synaptic dynamics, and feedforward synaptic dynamics.
Our model consists of three continuous-time dynamical systems evolving at three different time scales: neural dynamics, lateral synaptic dynamics, and feedforward synaptic dynamics at the fast, intermediate, and slow time scales, respectively.

By leveraging a multilevel optimization framework, we prove convergence of the dynamics in the offline setting.
Specifically, at the first level (fast time scale), in Theorem~\ref{lem:layer1} we showed strongly convexity of the cost function, explicitly computed the expression of the global minimum $Y^\star$ in equation~\eqref{eq:ystar}, and derived the corresponding continuous-time gradient-flow neural dynamics. By using contractivity we prove global exponential convergence to $Y^\star$ of the dynamics, along with other desirable properties of contracting dynamics.
Next, we derived similar results at the second level (intermediate time scale). Specifically, in Theorem~\ref{lem:layer2} we showed that the cost function at this level is strongly concave, derived the explicit expression of the global maximizer $M^\star$ in equation~\eqref{eq:Mstar} and the corresponding continuous-time gradient-flow lateral synaptic dynamics. We demonstrated that these dynamics exponentially converge to $M^\star$ in the space of positive definite matrices, leveraging contractivity arguments. Notably, we also proved a key invariance result for lateral synaptic dynamics: they remain positive definite and symmetric if initialized so.
Then, we analyzed the third and final level (slow time scale), which addressed the OP with respect to feedforward synaptic weights and involved a non-convex, non-smooth cost function. In Lemma~\ref{lem:stat_points}, we characterized the stationary points of the OP~\eqref{eq:op_layer3} within the set of full rank matrices. Next, in Theorem~\ref{lem:layer3}, we considered the corresponding continuous-time gradient-flow feedforward synaptic dynamics and showed almost everywhere convergence of this dynamics to a global minimum $W^\star$. These results relied upon two mild technical conjectures: (i) the forward invariance of the full rank set, ensuring the well-posedness of the dynamics~\eqref{eq:dyn_third_layer}; and (ii) the classification of the stationary points of~\eqref{eq:dyn_third_layer} as either minimizers or strict saddles, a property analytically challenging to verify. We empirically validated both conjectures in Appendix~\ref{apx:conjectures}.
We concluded the analysis by projecting the optimal solutions of the three levels back on the output space to obtain the final optimal $Y^\star$.
Finally, we validated the effectiveness of our results via numerical experiments.

In future work, we plan to extend this methodology of proof and our results to similarity-based settings for other dimensionality reduction problems~\cite{CP-DBC:19}. These including non-negative similarity matching for applications such as clustering and non-negative independent component analysis~\cite{CP-DBC:14, CP-SM-DBC:17, DL-CP-DBC:22}, as well as similarity matching with sparsity regularizer for, e.g., sparse dictionary learning~\cite{TH-CP-DBC:14}.
Additionally, with convergence in the offline setting established, a natural direction for further analysis is to extend our study to the online setting, aligning more closely with real-time neural computations.


\appendix
\section{Derivation of the Multi-Variable Optimization Problem~\eqref{eq:sm_offline}}
\label{apx:eqW_M}
For completeness, using the methodology from~\cite{CP-AMS-DBC:17}, here we provide the derivation of the multi-variable OP~\eqref{eq:sm_offline} starting from the similarity matching problem~\eqref{eq:similarity matching_matrix}. For this purpose, we give the following result.
\begin{lem}
\label{lem:eqW_M}
Given two matrices $Y \in \R^{m\times T}$ and $X \in \R^{n\times T}$ we have
\begin{align}
\label{eq:sm_W}
\min_{W \in \R^{m\times n}}
\normF{W}^2 -\frac{2}{T}\inprodF{WX, Y} & = -\frac{1}{T^2} \inprodF{\trasp{X}X, \trasp{Y}Y},\\
\label{eq:sm_M}
\max_{M \in \R^{m\times m}} 
\frac{2}{T} \inprodF{MY, Y} - \normF{M} &= \frac{1}{T^2} \normF{\trasp{Y}Y}^2.
\end{align}
\end{lem}
\begin{proof}
Let $F(W) := \normF{W}^2 -\frac{2}{T}\inprodF{WX, Y}  = -\frac{2}{T}\tr(\trasp{X}\trasp{W}Y) + \tr(\trasp{W}W)$ be the cost function of the minimization problem in equation~\eqref{eq:sm_W}. This function is continuous and strongly convex, so the OP~\eqref{eq:sm_W} admits a unique solution. To prove equality~\eqref{eq:sm_W}, we first find the optimal point of $F(W)$.
We have
\begin{align*}
\der{F(W)}{W} = 0 &\iff -\frac{2}{T}\der{\tr(\trasp{X}\trasp{W}Y)}{W} + \der{\tr(\trasp{W}W)}{W} = 0 \\
&\iff -\frac{2}{T} Y \trasp{X} + 2 W = 0 \iff W = \frac{1}{T} Y \trasp{X} := W^\star.
\end{align*}
Finally, we compute
\begin{align*}
{F(W^\star)} &= -\frac{2}{T^2}\tr(\trasp{X}X\trasp{Y}Y) + \frac{1}{T^2}\tr(X\trasp{Y} Y \trasp{X}) = -\frac{2}{T^2}\tr(\trasp{X}X\trasp{Y}Y) + \frac{1}{T^2}\tr(\trasp{X} X\trasp{Y} Y),
\end{align*}
where in the last equality we used the trace cyclic property. The proof of equality~\eqref{eq:sm_M} follows similar reasoning.
This concludes the proof.
\end{proof}
Note that the minimization problem in equation~\eqref{eq:sm_W} can be understood as minimizing the norm of the feedforward synaptic matrix $W$ while maximizing the similarity between the matrices $WX$ and $Y$. Similarly, the maximization problem in equation~\eqref{eq:sm_W} corresponds to minimizing the norm of the lateral synaptic matrix $M$ while maximizing the similarity between the matrices $MY$ and $Y$.
Next, using Lemma~\ref{lem:eqW_M}, the similarity matching optimization problem~\eqref{eq:sm_vec} can be embedded in a larger space as follows:
\begin{align*}
\min_{Y\in \R^{m \times T}} \frac{\tr(\trasp{Y}Y\trasp{Y}Y-2\trasp{X}X\trasp{Y}Y)}{T^2} =&\min_{Y\in \R^{m \times T}} \min_{W \in \R^{m\times n}} 2\normF{W}^2 - \frac{4}{T}\inprodF{WX, Y}\\
&+\min_{Y\in \R^{m \times T}}\max_{M \in \R^{m\times m}} \frac{2}{T}\inprodF{MY, Y} - \normF{M}^2.
\end{align*}
Changing the order of minimization with respect to $Y$ and $W$ as well as the order of minimization with respect to $Y$ and maximization with respect to $M$\footnote{The $\min-\max$ is justified by the saddle point property in~\cite[Appendix A, Proposition 1]{CP-AMS-DBC:17}.}, we finally get that the similarity matching problem~\eqref{eq:similarity matching_matrix} defined  can be embedded into the higher-dimensional space $\R^{m\times n} \times \R^{m\times m} \times \R^{m \times T}$ via the multi-variable OP
$$
\min_{W \in \R^{m\times n}} \max_{M \in \R^{m\times m}} \min_{Y\in \R^{m \times T}} 2\normF{W}^2 - \normF{M}^2 - \frac{4}{T}\inprodF{WX, Y} + \frac{2}{T}\inprodF{MY, Y}.
$$

\section{Proofs of the Results in Section~\ref{sec:convergence}}
\label{apx:proofs}
This appendix provides formal and detailed proofs of the results presented in Section~\ref{sec:convergence}.
\subsection{Proof of Theorem~\ref{lem:layer1}}
\label{apx:proof_lemma_layer1}
We first prove Theorem~\ref{lem:layer1}, which analyzes the minimization problem of the first level, namely problem~\eqref{eq:op_layer1}. 
\begin{proof}
First, we show part~\ref{item:1_first_layer}, i.e., the cost function $S_1(Y) = \wnormF{Y}{M}^2 - \frac{4}{T}\inprodF{WX, Y} + 2\normF{W}^2 - \normF{M}^2$ is strongly convex with parameter ${\nu}_{Y} = \frac{4}{T}\subscr{\lambda}{min}(M)$.
To do so, we show that the vectorized form of $S_1$ is strongly convex with parameter ${\nu}_{Y}$.  Let $y := \vecto(Y) \in \R^{mT}$ and note that 
\bi
\item $\wnormF{Y}{M}^2 = \tr(\trasp{Y}MY) = \tr(\trasp{Y}M^{1/2}M^{1/2}Y) = \tr(\trasp{(M^{1/2} Y)}M^{1/2}Y):=\normF{M^{1/2}Y}^2$,
\item $\normF{M^{1/2}Y}^2 = \norm{\vecto(M^{1/2}Y)}_2^2 \overset{\ref{prop:vec3}}{=} \norm{(I_{T} \otimes M^{1/2})\vecto(Y)}_2^2$,
\item $\tr(\trasp{(WX)}Y) \overset{\ref{prop:vec3}}{=} \trasp{\vecto(WX)}\vecto(Y)$.
\ei
The vectorized form of $S_1$ is therefore
$$
\begin{aligned}
S_1(y) &= \frac{2}{T}\Bigl(\trasp{y} \trasp{(I_{m} \otimes M^{1/2})}(I_{m} \otimes M^{1/2})y  - 2\trasp{\vecto(WX)}y \Bigr) + 2 \norm{\vecto(W)}_2^2 - \norm{\vecto(M)}_2^2.
\end{aligned}
$$
We have
\begin{align*}
& \nabla S_1(y) = \frac{2}{T} \Bigl(2\trasp{(I_{m} \otimes M^{1/2})}(I_{m} \otimes M^{1/2})y - 2\trasp{\vecto(WX)} \Bigr), \text{ and }\\
&\nabla^2 S_1(y) =  \frac{4}{T} \trasp{(I_{m} \otimes M^{1/2})}(I_{m} \otimes M^{1/2}) \succ \frac{4}{T} \subscr{\lambda}{min}(M),
\end{align*}
where the last inequality follows from $M \succ 0$ and the fact that the eigenvalues of the matrix $I_{m} \otimes M^{1/2}$ are the same (but with different multiplicity) of those of $M^{1/2}$. This concludes the proof of part~\ref{item:1_first_layer}.

Next, to prove part~\ref{item:2_first_layer}, we determine the minimizer of the OP~\eqref{eq:op_layer1}. We compute
\begin{align*}
\frac{d}{dY}\bigl(\wnormF{Y}{M}^2 - 2 \tr(\trasp{X}\trasp{W}Y)\bigr) := \frac{d}{dY}\bigl(\tr(\trasp{Y}MY) - 2 \tr(\trasp{X}\trasp{W}Y)\bigr) &= 0 \iff \\
2 MY - 2 WX &= 0 \iff Y = M^{-1}WX =: Y^\star.
\end{align*}
This concludes the proof of part~\ref{item:2_first_layer}.
Finally, part~\ref{item:3_first_layer} follows directly from~\ref{item:1_first_layer} by applying Lemma~\ref{lem:gradient_flow}, after computing $\nabla_Y S(W, M, Y)$.
\end{proof}

\subsection{Proof of Theorem~\ref{lem:layer2}}
\label{apx:proof_lemma_layer2}
We now prove Theorem~\ref{lem:layer2}, which analyzes the OP of the second level, i.e., problem~\eqref{eq:op_layer2}. 
\begin{proof}
To prove part~\ref{item:1_second_layer}, i.e., the cost function $S_2$ is strongly concave with parameter ${\nu}_{M} =  2$, we show that the function $M \mapsto - S_2(M) - {\frac{2}{2}}\normF{M}^2$ is convex\footnote{A function $f$ is strongly convex with parameter $\nu$ if and only if the function $x\mapsto f(x)-{\frac {\nu}{2}}\|x\|^{2}$ is convex.}.
To this purpose, first note that the function $M \mapsto \normF{M}^2$ is convex, being a norm. Additionally, the map $M \mapsto \tr{(M^{-1}W C_X\trasp{W})}$ is convex. 
This property follows from the following reasoning. Consider a small perturbation of $M$, that is $M(t) = M + tP$, where $P \in \S^m$. Then, it suffices to prove that $\frac{d^2}{dt^2} \tr(M^{-1}W C_X\trasp{W})\big |_{t = 0} \geq 0$.
We compute
\begin{align*}
M(t)^{-1}W C_X\trasp{W} &= \bigl(M + tP\bigr)^{-1}W C_X\trasp{W} = \bigl(M(I_m + tM^{-1}P)\bigr)^{-1}WC_X\trasp{W} \\
&= \bigl(I_m + tM^{-1}P\bigr)^{-1}M^{-1}WC_X\trasp{W} \\
&= M^{-1}WC_X\trasp{W} - tM^{-1}PM^{-1}WC_X\trasp{W} + t^2 M^{-1}PM^{-1}PM^{-1}WC_X\trasp{W} + O(t^3),
\end{align*}
where in the last equality we expanded $(I_m + tM^{-1}P)^{-1}$ in power series.
We have
\[
\frac{d^2}{dt^2} \tr(M^{-1}W C_X\trasp{W})\big|_{t = 0} = 2 \tr(M^{-1}PM^{-1}PM^{-1}WC_X\trasp{W}) = \inprodF{(M^{-1}P)M^{-1}(M^{-1}P)^\top,WC_X\trasp{W}} \geq 0,
\]
where the last inequality follows from $M^{-1} \succ 0$ and $C_X \succ 0$. Recalling that any non-negative weighted sum of convex functions is convex, we have that $\inprodF{M^{-1}, WC_X\trasp{W}} $ is convex.
Then, the map
$$
M \mapsto -S_2(M) - {\frac{2}{2}}\normF{M}^2 = \normF{M}^2 + 2\inprodF{M^{-1}, WC_X\trasp{W}} - \normF{M}^2 = 2\inprodF{M^{-1}, WC_X\trasp{W}}
$$
is convex. This concludes the proof of part~\ref{item:1_second_layer}. 
To prove~\ref{item:2_second_layer}, we compute
\begin{align*}
\frac{d}{dM}\bigl({\tr(\trasp{M}M) + 2 \tr(M^{-1}W C_X\trasp{W})\bigr)} = 0 &\iff 2 M - 2 M^{-1}W C_X\trasp{W}M^{-1}  = 0\\
& \iff M^3 = W C_X\trasp{W} \iff M = \bigl(WC_X\trasp{W}\bigr)^{1/3} =: M^\star.
\end{align*}

Next, to prove part~\ref{item:3_second_layer},  let $\subscr{M}{S}(t)$ and $\subscr{M}{A}(t)$ be the symmetric and skew-symmetric components of $M(t)$, respectively.
To prove~\ref{item1:lem_M} we show that $M(t) = \subscr{M}{S}(t)$, $\forall t \geq 0$.
By using the standard decomposition $M(t) = \subscr{M}{S}(t) + \subscr{M}{A}(t)$, $\forall t \geq 0$, we can write equation~\eqref{eq:dynamics_M} as
$$
\subscr{\dot M}{S}(t) + \subscr{\dot M}{A}(t) = - \subscr{M}{S}(t) + \frac{2}{T} Y(t) \trasp{Y(t)} - \subscr{M}{A}(t).
$$
Note that the right hand side in equality~\eqref{eq:dynamics_M} is symmetric at time $t=0$,
being (i) $Y(t) \trasp{Y(t)}$ symmetric for all $t\geq0$, and (ii) $\subscr{M}{A}(0) = 0$ for assumption. Thus $\subscr{M}{A}(t) = 0$, $\forall t\geq 0$. Additionally $\ds \lim_{t \to \infty} \subscr{M}{A}(t) = 0$. This equality follows by noticing that the dynamics for the skew-symmetric component of $M(t)$ are given by $\subscr{\dot M}{A}(t) = -\subscr{M}{A}(t)$, whose solution is $\subscr{M}{A}(t) = \subscr{M}{A}(0)\e^{-t}$ = 0, for all $t \geq 0$. Thus the desired result.

To prove~\ref{item2:lem_M} we use the definition of positive definite matrix, i.e., $M\succ 0$ if and only if $\trasp{x}M x > 0$, for all vectors $x$ different from zero. 
Given $z \in \R^{m}\setminus{\0_m}$, consider the scalar dynamics 
\beq
\label{eq:dynamics_M_x}
\frac{d}{dt} \trasp{z}M(t)z = - \trasp{z}M(t)z + \trasp{z}\frac{Y(t) \trasp{Y(t)}}{T}z.
\eeq
The solution of the ODE~\eqref{eq:dynamics_M_x} is
$
\trasp{z}M(t)z = \bigl(\trasp{z}M_0z\bigr)\e^{-t} + \int_0^t \e^{\tau -t} \Bigl(\trasp{z}\frac{Y(\tau) \trasp{Y(\tau)}}{T}z\Bigr) d\tau
$
and, being $M_0\succ0$, we have $\trasp{z}M(t)z > 0$ for all $t \geq 0$. This concludes the proof.

Finally, statement~\ref{item:4_second_layer} follows by computing
$
\nabla_M S_2(M) = \nabla_M S(W, M, Y)\big|_{Y^\star} {=} - 2M + \frac{2}{T}Y^\star (Y^\star)^\top = - 2 M + 2 M^{-1}W C_X\trasp{W}M^{-1},
$
and by applying Lemma~\ref{lem:gradient_flow} within the set $\S_{\succ 0}^m$, which, by statement~\ref{item:3_second_layer} is forward invariant for the dynamics~\eqref{eq:dynamics_M}.
\end{proof}
\subsection{Proof of Lemma~\ref{lem:stat_points}}
\label{apx:proof_lemma_stat_points}
We prove that the stationary points of the minimization problem~\eqref{eq:op_layer3} are the full rank matrices whose eigenvalues and associated right eigenvectors are $m$ eigenvalues of $C_X$ with the associated eigenvectors.
\begin{proof}
For any $W \in \mcR_{m,n}$, we have $\nabla S_3(W) = - 4W + 4\bigl(WC_X\trasp{W}\bigr)^{-1/3} WC_X$. To prove our statement we show that $\nabla S_3(\bar{W}) = 0$ if and only if $\bar{W} \in \mathcal{S}$. Let $(U_W, \Sigma_W, V_W)$ be the SVD of $W \in \R^{m \times n}$. We have 
\begin{align}
&- 4W + 4\bigl(WC_X\trasp{W}\bigr)^{-1/3} WC_X = \0_{m \times n} \nonumber\\
& \iff - U_W \Sigma_W V_W^\top + \bigl(U_W \Sigma_W V_W^\top C_XV_W \Sigma_W^\top U_W^\top \bigr)^{-1/3} U_W \Sigma_W V_W^\top C_X = \0_{m \times n} \label{eq:1}\\
& \iff - U_W \Sigma_W V_W^\top + U_W\bigl(\Sigma_W V_W^\top C_XV_W \Sigma_W^\top \bigr)^{-1/3} U_W^\top U_W \Sigma_W V_W^\top C_X = \0_{m \times n}\label{eq:2}\\
& \iff \Sigma_W = \bigl(\Sigma_W V_W^\top C_XV_W \Sigma_W^\top \bigr)^{-1/3} \Sigma_W V_W^\top C_X V_W, \label{eq:3}
\end{align}
where in equality~\eqref{eq:1} we substituted the SVD composition of $W$; in equality~\eqref{eq:2} we used the definition of exponential matrix; in equality~\eqref{eq:3} we multiplied on the left by $U_W^\top$ and on the right by $V_W$ and used the fact that $U \in \mathcal{O}_m$ and $V \in \mathcal{O}_n$.
Next, being $\Sigma_W \in \mathcal{R}_{m,n}$, equality~\eqref{eq:3} is equivalent to
\begin{align}
&\Sigma_W \Sigma_W^\top = \bigl(\Sigma_W V_W^\top C_XV_W \Sigma_W^\top \bigr)^{-1/3} \Sigma_W V_W^\top C_X V_W \Sigma_W^\top \iff (\Sigma_W^m)^2 = \bigl(\Sigma_W V_W^\top C_XV_W \Sigma_W^\top \bigr)^{2/3} \label{eq:5}\\
&\iff (\Sigma_W^m)^3 = \Sigma_W^m (V_W^\top V_C \Lambda_C V_C^\top V_W)_{(m,m)}(\Sigma_W^m)^\top \label{eq:6}\\
&\iff
\left[
\begin{array}{c c c}
\sigma_1^3 & \dots & 0\\
\vdots & \ddots & 0\\
0 & \dots & \sigma_m^3
\end{array}
\right] = 
\left[
\begin{array}{c c c}
\sigma_1 & \dots & 0\\
\vdots & \ddots & 0\\
0 & \dots & \sigma_m
\end{array}
\right]
\left(
V_W^\top V_C
\left[
\begin{array}{c c c}
\lambda_1^C & \dots & 0\\
\vdots & \ddots & 0\\
0 & \dots & \lambda_n^C
\end{array}
\right]
V_C^\top V_W
\right)_{(m,m)}
\left[
\begin{array}{c c c}
\sigma_1 & \dots & 0\\
\vdots & \ddots & 0\\
0 & \dots & \sigma_m
\end{array}
\right] \nonumber\\
&\iff
\left[
\begin{array}{c c c}
\sigma_1 & \dots & 0\\
\vdots & \ddots & 0\\
0 & \dots & \sigma_m
\end{array}
\right] = 
\left(
V_W^\top V_C
\left[
\begin{array}{c c c}
\lambda_1^C & \dots & 0\\
\vdots & \ddots & 0\\
0 & \dots & \lambda_n^C
\end{array}
\right]
V_C^\top V_W
\right)_{(m,m)}, \label{eq:7}
\end{align}
where in equality~\eqref{eq:5} we multiplied on the right by $\Sigma_W^\top$; in equality~\eqref{eq:6} we raised both sides to the $2/3$ power.
We stress that equation~\eqref{eq:7} has two unknowns: the singular values $\sigma_i$, $i \in \until{m}$, and the corresponding right eigenvectors $V_W$.
Being the LHS of equality~\eqref{eq:7} diagonal, this equality is verified when $V_C V_W^\top$ is verified if and only if $V_C V_W^\top \in \mathcal{P}_n$. Then the points satisfying equality~\eqref{eq:7} are the matrices $W$ such that (i) $V_W = V_C P^\top$, (ii) $\Sigma_W^m = \left( P \Lambda_C P^\top \right)_{(m,m)}$, for $P \in \mathcal{P}_n$.
This concludes the proof.
\end{proof}
\subsection{Proof of Lemma~\ref{lem:coercive}}
\label{apx:proof_lemma_lcoercive}
We prove that the cost function $S_3(W) = 2\normF{W}^2 - 3\normF{\bigl(WC_X\trasp{W}\bigr)^{1/3}}^2$ of the OP~\eqref{eq:op_layer3} is coercive. This follows noticing that the first term grows as $\bigO(\normF{W}^2)$, while the second term grows only as $\bigO(\normF{W}^{4/3})$, being $C_X \succ 0$. Therefore, as $\normF{W}^2$ goes to infinity, the quadratic term dominates, and $S_3(W) \to + \infty$.
\subsection{Proof of Theorem~\ref{lem:layer3}}
\label{apx:proof_lemma_layer3}
We now prove Theorem~\ref{lem:layer3}, which analyzes the OP of the third and last level, namely problem~\eqref{eq:op_layer3}, characterizing convergence of the continuous-time gradient-flow feedforward synaptic dynamics in the slow time scale~\eqref{eq:dyn_third_layer}.
\begin{proof}
Under Conjecture~\ref{conj_full_row_rank}, LaSalle’s invariance principle implies that any trajectory $W(t)$ of the dynamics~\eqref{eq:dyn_third_layer} converges to the set $\mathcal{S}$ of full-rank critical points introduced in Lemma~\ref{lem:stat_points}.
Next, given the SVD decomposition of $W$, we note that (i) the dynamics is invariant with respect to the left eigenvectors $U_W$, and (ii) at fixed $U_W$, the critical points $\bar{W} \in \mathcal{S}$ are isolated. Therefore, for each trajectory $W(t)$ of the dynamics~\eqref{eq:dyn_third_layer}, there exists a unique critical point $\bar{W} \in \mathcal{S}$ to which $W(t)$ converges.
Finally, Conjecture~\ref{conj_full_strict_saddle} ensures that all non-minimizing critical points are strict saddles, while Lemma~\ref{lem:coercive} ensures that any trajectory of the gradient-flow~\eqref{eq:dyn_third_layer} is bounded. Then, convergence to a global minimizer $W^\star$ from almost all initial conditions follows by applying~\cite[Corollary 4]{ACBdO-MS-EDS:24}.
\end{proof}

\subsection{Proof of Lemma~\ref{lem:final_proj}}
\label{apx:proof_lemma_final_proj}
Finally, we show that projecting the optimal solutions of the three levels back into the output space, we obtain the optimal $Y^\star$ given in Lemma~\ref{lem:exact_sol_sm}.
\begin{proof}
Let $(U_X, \Sigma_X, V_X)$ be the SVD of $X \in \R^{n \times T}$, and $(V_C, \Lambda_C)$ be the SVD of the covariance matrix $C_X \in \R^{n \times n}$. By construction, $V_C = U_X \in \R^{n \times n}$ and $\ds T\Lambda_C = \Sigma_X \Sigma_X^\top \in \R^{n \times n}$.
Then, $W^\star =   U \Lambda_C^m (V_C^m)^\top = U \Lambda_C^m (U_X^m)^\top$, where $U\in \R^{m \times m}$ is an arbitrary orthogonal matrix and $U_X^m = (u^X_1, \dots, u^X_m)$. 
Thus 
$
W^\star C_X (W^\star)^\top = U(\Lambda_C^m)^3 U^\top.
$
Projecting back the optimal solutions of the three levels on the output space
we have
\begin{align}
Y^\star(M^\star, W^\star) &\overset{\eqref{eq:ystar}}{=}  (M^\star)^{-1} W^\star X \overset{\eqref{eq:Mstar}}{=} \bigl(U(\Lambda_C^m)^3U^\top\bigr)^{-1/3} U \Lambda_C^m(U_X^m)^\top X
\nonumber\\
&= U (\Lambda_C^m)^{-1} \Lambda_C^m (U_X^m)^\top U_X \Sigma_X V_X^\top
\label{eq:3_orth_svd}\\
&= U I_m (U_X^m)^\top
\left[
\begin{array}{c}
U_X^{m}\\
U_X^{n-m}
\end{array}
\right] \Sigma_X V_X^\top = U \Sigma_X^m (V_X^m)^\top, \label{eq:4_orth_prod}
\end{align}
where in equality~\eqref{eq:3_orth_svd} we used 
the orthogonality of $U$ and substituted the SVD of $X$, while in equality~\eqref{eq:4_orth_prod} we used the orthogonality of $U_X$.
\end{proof}
\section{Empirical Evidence for Conjectures in Section~\ref{sec:third_layer}}
\label{apx:conjectures}
This appendix presents numerical evidence in support of Conjectures~\ref{conj_full_row_rank} and~\ref{conj_full_strict_saddle}. We validate both conjectures on the same set of numerical simulations. Specifically, we consider three different $(m, n)$ pairs: $(2,10)$, $(3,30)$, and $(5,50)$. For each pair, we randomly generate a positive definite matrix $C_X$ and run $10.000$ simulations of the dynamics~\eqref{eq:dyn_third_layer} over the time interval $[0, 30]$, starting from a full rank randomly generated $W_0$ for each run. The code and the data to reproduce all the simulations in this section are available at the GitHub repository~\url{https://shorturl.at/l1qVf}.

\subsection{Empirical Evidence for Conjecture~\ref{conj_full_row_rank}}
\label{apx:conjecture1}
Conjecture~\ref{conj_full_row_rank} asserts that the trajectories $W(t)$ of the dynamics~\eqref{eq:dyn_third_layer} remain full rank for all time, if $W_0$ has full row rank.
To test this, we track the rank of $W(t) W(t)^\top$ over time across $10.000$ independent runs for each $(m,n)$ pair.
Figure~\ref{fig:rank_dynamics} shows the evolution of the rank over time for each pair. In agreement with our conjecture, the rank remains constant at $m$ throughout the entire simulation for all tested cases.
\begin{figure}[!h]
\centering
\includegraphics[width=.8\linewidth]{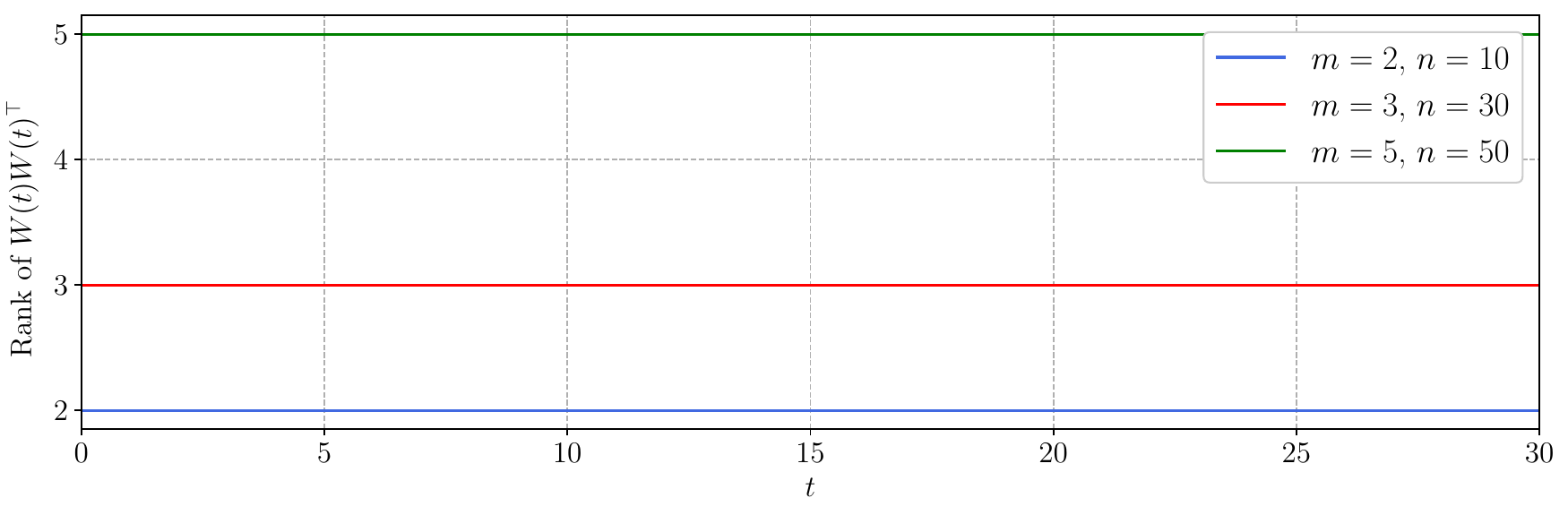}
\caption{Rank of $W(t) W(t)^\top$ over time across the $10.000$ independent runs for three different $(m, n)$ pairs: $(2,10)$, $(3,30)$, and $(5,50)$. The matrix $C_X \succ 0$ is the same for each pair. The dynamics~\eqref{eq:dyn_third_layer} are simulated over the time interval $[0, 30]$, starting from full rank randomly generated $W_0$. In agreement with our conjecture, the rank is constant at $m$ over the simulations ($t$ is simulation time). 
}
\label{fig:rank_dynamics}
\end{figure}

\subsection{Empirical Evidence for Conjecture~\ref{conj_full_strict_saddle}}
\label{apx:conjecture2}
To test Conjecture~\ref{conj_full_strict_saddle}, we compute the numerical equilibrium points $W^\star$ reached by the dynamics~\eqref{eq:dyn_third_layer} in each run. Then, we compute the average of the singular values and right eigenvectors of the matrix $W^\star$ across all 10.000 simulations. 
For the sake of space, Table~\ref{tab:check_conjecture2} reports only the average of the singular values of $W^\star$ and the largest six eigenvalues of the covariance matrix $C_X$. In accordance with our conjecture, $W^\star$ is the matrix having as singular values the largest $m$ eigenvalues of $C_X$, and as right eigenvectors the associated eigenvectors of $C_X$.
\begin{table}[!h]
\centering
\begin{tabular}{ccc}
\hline
m &  $[\sigma_1^W, \dots, \sigma_m^W]$ & $[\lambda_1^C, \dots, \lambda_6^C]$ \\
\hline
  2 & [166.57, 131.05]  & [166.57, 131.06, 120.74, 107.36, 98.86, 93.53] \\
  3 & [231.01, 199.67, 181.57]  & [231.01, 199.67, 181.59, 172.83, 157.73, 151.66]  \\
  5 & [172.97, 161.13, 157.47, 155.18, 150.15] & [172.97, 161.13, 157.47, 155.18, 150.26, 145.76] \\
\hline
\end{tabular}
\caption{Comparison between the average $m$ singular values of the matrix $W^\star$ (second column) and the top six eigenvalues of the randomly generated matrix $C_X$ (third column). Each row corresponds to a different $(m, n)$ pair: $(2,10)$, $(3,30)$, and $(5,50)$. The averages of $\sigma^W$ are computed over 10.000 simulations. In agreement with Conjecture~\ref{conj_full_strict_saddle}, the eigenvalues of $W^\star$ are the leading $m$ eigenvalues of $C_X$.}
\label{tab:check_conjecture2}
\end{table}
Next, we consider an example with $n,~m,~T = 5, 3, 10$, where $C_X \succ 0$ is randomly generated and has eigenvalues $\Lambda_C = [3.8, 2.4, 2.2, 1.5, 1.2]$. Starting from a global minimum $W^\star$ as characterized by Conjecture~\ref{conj_full_strict_saddle}, we evaluate the cost function $S_3(W)$ by varying one singular value $\sigma_i^W$ while keeping the others fixed. Figure~\ref{figure:costs} illustrates the results. In each case, the minimum is obtained when the singular value $\sigma_i^W$ matches the corresponding eigenvalue $\lambda_i^C$. Additionally, we also consider a case where the first and last rows of $V_C$ are swapped. For this case, we set $\tilde W = \Sigma_W V_C^\top$ and again evaluate $S_3(W)$ as a function of $\sigma_1^W$. In both cases, the results confirm that the global minimum of $S_3(W)$ is obtained at $W^\star$.
\begin{figure}[!h]
\centering
\includegraphics[width=.8\linewidth]{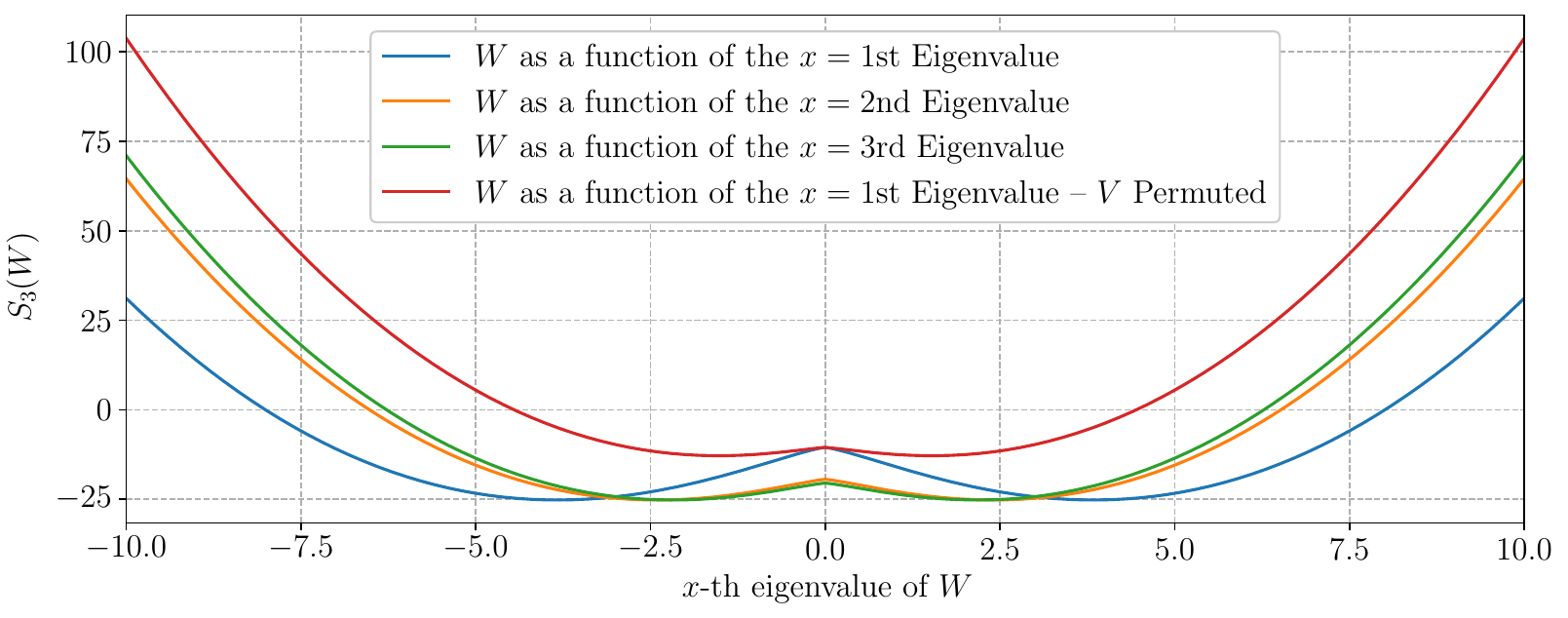}
\caption{Cost function $S_3(W)$ as a function of one singular value $\sigma_i^W$, while the others are held fixed at their values in $W^\star = \Sigma_C^m V_C^\top$.The plot also shows the case where the first and last rows of $V_C$ are permuted. In the simulations, we set $n,~m,~T = 5, 3, 10$ and $C_X\succ 0$ is randomly generated and has eigenvalues $\Lambda_C = [3.8, 2.4, 2.2, 1.5, 1.2]$. The minimum of $S_3(W)$ is reached at $W^\star$, in agreement with Conjecture~\ref{conj_full_strict_saddle}.}
\label{figure:costs}
\end{figure}
\section{Alternative Approach for Computing the Minimizer of Problem~\eqref{eq:op_layer3}}
In this section, we present an alternative method for computing the minimizer of the optimization problem corresponding to the third level, namely Problem~\eqref{eq:op_layer3}. This approach, which characterizes the landscape of the minimizers, is interesting {\em per se} and is independent on the dynamics~\eqref{eq:dyn_third_layer}.
We begin with an algebraic result on the maximization of a power of the trace of diagonal and orthonormal matrices product that is instrumental for our analysis.
It is known that, given $\Lambda_1 \in \R^{m \times n}$, and $\Lambda_2 \in \R^{n \times n}$, $\Lambda_2 \succ 0$, diagonal matrices, 
$\ds \max_{V \in \mathcal{O}_n} \tr\bigl(\Lambda_1 V \Lambda_2 V^\top \Lambda_1^\top\bigr) = \tr\bigl(\Lambda_1 \Lambda_2 \Lambda_1^\top \bigr)$. 
In Lemma~\ref{lem:pita} we extend this result to the case where the cost function is raised to the $2/3$ power. Specifically, we show that $\ds \max_{V \in \mathcal{O}_n} \tr\bigl(\bigl(\Lambda_1 V \Lambda_2 V^\top \Lambda_1^\top\bigr)^{2/3}\bigr) = \tr\bigl(\bigl(\Lambda_1 \Lambda_2 \Lambda_1^\top\bigr)^{2/3}\bigr)$. In our formal derivations, a mild technical point in the proof relies on the following conjecture:
\begin{conj}
\label{conj_lin_prog}
Consider $\Lambda_1 \in \R^{m \times n}$ and $\Lambda_2 \in \R^{n \times n}$ diagonal matrices, with $\Lambda_2 \succ 0$ with distinct diagonal entries, and let $U$ be an orthogonal matrix. If the matrix $U^\top \Lambda_1^\top\bigl(\Lambda_1 U \Lambda_2 U^\top \Lambda_1^\top\bigr)^{-1/3}\Lambda_1 U$ is diagonal, then $U$ is a permutation matrix.
\end{conj}
Note that the converse of Conjecture~\ref{conj_lin_prog} always holds. That is, if $U \in \mathcal{P}_m$, then $U^\top \Lambda_1^\top \bigl(\Lambda_1 U \Lambda_2 U^\top \Lambda_1^\top\bigr)^{-1/3} \Lambda_1 U$ is diagonal.
To empirically validate our conjecture, we conduct extensive numerical tests by running $100000$ simulations for each of five different dimension pairs $(m, n)$: $(5, 10)$, $(10, 20)$, $(10, 50)$, $(50, 100)$, and $(50, 500)$. In each simulation, we randomly generated an orthogonal matrix $U$ and diagonal matrices $\Lambda_1$ and $\Lambda_2$.
Then, we numerically computed $\Lambda_3$ and checked whether this matrix was diagonal. Even though it is unlikely for a randomly generated orthogonal matrix $U$ to be a permutation matrix, we checked it in each simulation and, indeed, found that $U$ was never a permutation matrix. The results show that, in agreement with our conjecture, the matrix $\Lambda_3$ is not diagonal when $U \notin \mathcal{P}_m$. The code to replicate the simulations is available at the GitHub~\url{https://shorturl.at/l1qVf}.
Next, we give the following instrumental result.
\begin{lem}[Maximization of the trace of diagonal and orthonormal matrices product]
\label{lem:pita}
Consider $\Lambda_1 \in \R^{m \times n}$, and $\Lambda_2 \in \R^{n \times n}$, $\Lambda_2 \succ 0$, diagonal matrices\footnote{For a matrix $\Lambda \in \R^{m \times n}$, $m\leq n$, by diagonal we mean a matrix of the form $\left[
\begin{array}{c|c}
\Lambda^m & 0_{m, n-m}
\end{array}
\right]$, with $\Lambda^m \in \R^{m \times m}$ diagonal.} with distinct diagonal entries ordered in descending order. If Conjecture~\ref{conj_lin_prog} holds, then,
\beq
\label{eq:pita}
\max_{V \in \mathcal{O}_n} \tr\bigl(\bigl(\Lambda_1 V \Lambda_2 V^\top \Lambda_1^\top\bigr)^{2/3}\bigr) = \tr\bigl(\bigl(\Lambda_1 \Lambda_2 \Lambda_1^\top\bigr)^{2/3}\bigr).
\eeq
\end{lem}
\begin{proof}
First, note that the trace is a continuous function and the set of all real orthogonal matrices is compact, hence the maximum is well-defined.
Let $\map{g}{\mathcal{O}_n}{\R}$ be the objective function of the maximization problem in equation~\eqref{eq:pita}, defined by $g(V) := \tr\bigl(\bigl(\Lambda_1 V \Lambda_2 V^\top \Lambda_1^\top\bigr)^{2/3}\bigr)$. 
We recall that the matrix exponential of every skew-symmetric matrix is a real orthogonal matrix of determinant one\footnote{Recall that for any matrix $A$, the identity $\det \left(\e^{A}\right)=\e^{\tr(A)}$ holds. Additionally, the diagonal elements of any skew-symmetric matrix are zero, and therefore its trace is zero.}.
Let $U$ be a global maximizer of problem~\eqref{eq:pita}, then for any skew-symmetric matrix $S$, we have $g(U) \geq g(U e^{tS})$, for all $t$.
This implies that, defining $f(t) :=\tr \bigl(\bigl(\Lambda_1 U \e^{tS} \Lambda_2 \e^{-tS} U^\top \Lambda_1^\top\bigr)^{2/3}\bigr)$, we must have $f'(0) = 0$. Note that the trace operator is linear, so for any time-varying matrix $A(t)$, it holds $\ds \frac{d}{dt} \tr(A(t)) = \ds \tr \Bigl(\frac{d}{dt} A(t)\Bigr)$.
We compute
\begin{align}
f'(t) &= \tr\Bigl(\frac{2}{3}\bigl(\Lambda_1 U \e^{tS} \Lambda_2 \e^{-tS} U^\top \Lambda_1^\top\bigr)^{-1/3}\bigl( \Lambda_1 U S\e^{tS} \Lambda_2 \e^{-tS} U^\top \Lambda_1^\top - \Lambda_1 U \e^{tS} \Lambda_2 S\e^{-tS} U^\top \Lambda_1^\top \bigr)\Bigr)\nonumber\\
&= \tr\Bigl(\frac{2}{3}\bigl(\Lambda_1 U \e^{tS} \Lambda_2 \e^{-tS} U^\top \Lambda_1^\top\bigr)^{-1/3}\Lambda_1 U S\e^{tS} \Lambda_2 \e^{-tS} U^\top \Lambda_1^\top\Bigr)\label{eq:df_1}\\
& \quad- \tr\Bigl(\frac{2}{3}\bigl(\Lambda_1 U \e^{tS} \Lambda_2 \e^{-tS} U^\top \Lambda_1^\top\bigr)^{-1/3}\Lambda_1 U \e^{tS} \Lambda_2 S\e^{-tS} U^\top \Lambda_1^\top\Bigr)\nonumber\\
&= \tr\Bigl(\frac{2}{3}\e^{tS} \Lambda_2 \e^{-tS} U^\top \Lambda_1^\top\bigl(\Lambda_1 U \e^{tS} \Lambda_2 \e^{-tS} U^\top \Lambda_1^\top\bigr)^{-1/3}\Lambda_1 U S\Bigr) \label{eq:df_2}\\
& \quad- \tr\Bigl(\frac{2}{3}\e^{-tS} U^\top \Lambda_1^\top\bigl(\Lambda_1 U \e^{tS} \Lambda_2 \e^{-tS} U^\top \Lambda_1^\top\bigr)^{-1/3}\Lambda_1 U \e^{tS} \Lambda_2 S\Bigr),
\nonumber
\end{align}
where in equality~\eqref{eq:df_1} 
we used the trace additivity property, while in equality~\eqref{eq:df_2} we used the cyclic property of the trace. We have
\[
0 = f'(0) = \tr\Bigl(\frac{2}{3}\Bigl(\Lambda_2 U^\top \Lambda_1^\top\bigl(\Lambda_1 U \Lambda_2 U^\top \Lambda_1^\top\bigr)^{-1/3}\Lambda_1 U - U^\top \Lambda_1^\top\bigl(\Lambda_1 U\Lambda_2 U^\top \Lambda_1^\top\bigr)^{-1/3}\Lambda_1 U\Lambda_2\Bigr)S\Bigr).
\]

Now, letting $S^\top := \Lambda_2 U^\top \Lambda_1^\top\bigl(\Lambda_1 U \Lambda_2 U^\top \Lambda_1^\top\bigr)^{-1/3}\Lambda_1 U - U^\top \Lambda_1^\top\bigl(\Lambda_1 U\Lambda_2 U^\top \Lambda_1^\top\bigr)^{-1/3}\Lambda_1 U\Lambda_2$ (which is skew-symmetric), we get $\tr(S^\top S) = 0$. 
This in turn implies $S = \0_{n\times n}$, thus
\beq
\label{eq:commutes}
\Lambda_2 U^\top \Lambda_1^\top\bigl(\Lambda_1 U \Lambda_2 U^\top \Lambda_1^\top\bigr)^{-1/3}\Lambda_1 U = U^\top \Lambda_1^\top\bigl(\Lambda_1 U\Lambda_2 U^\top \Lambda_1^\top\bigr)^{-1/3}\Lambda_1 U\Lambda_2.
\eeq
Equality~\eqref{eq:commutes} implies that the matrix $U^\top \Lambda_1^\top\bigl(\Lambda_1 U \Lambda_2 U^\top \Lambda_1^\top\bigr)^{-1/3}\Lambda_1 U$ commutes with the matrix $\Lambda_2$. Note that any matrix that commutes with a diagonal matrix with distinct diagonal entries must have all off-diagonal entries equal to zero. Therefore, the matrix 
$
\Lambda_3 := U^\top \Lambda_1^\top\bigl(\Lambda_1 U \Lambda_2 U^\top \Lambda_1^\top\bigr)^{-1/3}\Lambda_1 U
$
is diagonal. Then, according to Conjecture~\ref{conj_lin_prog}, $U$ is a permutation matrix. 
In summary, we have proved that the stationary points of the objective function of problem~\eqref{eq:pita} are obtained when $U$ is a permutation matrix.

Returning to our maximization problem~\eqref{eq:pita}, consider the matrix product $\Lambda_1 V \Lambda_2 V^\top \Lambda_1^\top$. Being $V \in \R^{n \times n}$ a permutation matrix, the product $V \Lambda_2 V^\top$ is a diagonal matrix whose diagonal entries, say them $\tilde \lambda_2$, are a permutation of the diagonal elements of $\Lambda_2$. Additionally, recall that the matrix $\Lambda_1 \in \R^{m \times n}$ is of the form
\[
\Lambda_1 =
\begin{pmatrix}
\Lambda_1^m & \0_{m, n-m}
\end{pmatrix},
\]
where $\Lambda_1^m = [\lambda_1] \in \R^{m \times m}$.
Therefore the product $\Lambda_1 V \Lambda_2 V^\top \Lambda_1^\top$ is an $m\times m$ diagonal matrix with the $i$-th diagonal entry given by the scalar $\lambda_{1,i}^2 \tilde \lambda_{2,i}$, for $i \in \until{m}$.
As a consequence, we have 
$$
\tr\bigl(\bigl(\Lambda_1 V \Lambda_2 V^\top \Lambda_1^\top\bigr)^{2/3}\bigr) = \tr\bigl(\bigl(\Lambda_1P\Lambda_2 P^\top\Lambda_1^\top\bigr)^{2/3}\bigr) = \sum_{i = 1}^m (\lambda_{1,i}^2 \tilde \lambda_{2,i})^{2/3},
$$
where $P \in \mathcal{P}_n$. Since we are assuming that the elements of $\lambda_1$ and $\lambda_2$ are in descending order, then, the $n \times n$ permutation matrices that solve~\eqref{lem:pita} are those keeping the largest $m$ component of the vector $\lambda_2$ in the same positions. That is, are the matrices of the form
\[
P^\star=
\left[
\begin{array}{c|c}
I_m & \0_{m, n-m} \\
\hline
\0_{n-m, m} & P_{n-m}
\end{array}
\right],
\]
where $P_{n-m}$ is a permutation matrix of size $n-m$.
This concludes the proof.
\end{proof}

With Lemma~\ref{lem:pita} in mind, we now proceed to the analysis of the OP at the third level.
\bt[On the minima of~\eqref{eq:op_layer3}]
\label{lem:layer3_alternative}
Consider the minimization problem~\eqref{eq:op_layer3}. Assume Conjecture~\ref{conj_lin_prog} holds. For any $X \in \R^{n \times T}$ such that $C_X \succ 0$, the solution of the minimization problem~\eqref{eq:op_layer3} is
\beq
\label{eq:Wstar}
W^\star = U \Lambda_C^m (V_C^m)^\top \in \R^{m \times n},
\eeq
where $U\in \R^{m \times m}$ is an arbitrary orthogonal matrix, $(V_C, \Lambda_C)$ is the SVD of the covariance matrix $C_X \in \R^{n \times n}$, with $\Lambda_C^m = \diag{\lambda_1^C, \dots, \lambda_m^C}$, and $V_C^m = (v^C_1, \dots, v^C_m)\in \R^{n\times m}$.
\et
\begin{proof}
To solve the minimization problem~\eqref{eq:op_layer3}, we make a change of variables and analyze an equivalent problem. 
Specifically, let $W = r \tilde W$, with $r >0$, and $\tilde W \in \R^{m \times n}$ satisfying $\normF{\tilde W} = 1$. Problem~\eqref{eq:op_layer3} is then equivalent to 
\begin{align}
\label{eq:sm_offline_norms_w_a}
\min_{r>0,~\tilde W \in \R^{m\times n}} \Bigl(2r^2 - 3r^{4/3}\normF{\bigl(\tilde WC_X\trasp{\tilde W}\bigr)^{1/3}}^2\Bigr)
&= \min_{r>0} \Bigl(2r^2 - 3r^{4/3}\max_{\tilde W, \normF{\tilde W} = 1}\normF{\bigl(\tilde WC_X\trasp{\tilde W}\bigr)^{1/3}}^2\Bigr),
\end{align}
where we used the well-known equality $\min(- f(x)) = - \max (f(x))$. To solve the OP~\eqref{eq:sm_offline_norms_w_a} we proceed as follows:
\begin{enumerate}[label=(\arabic*)]
\item find $\ds \tilde W^\star = \argmax_{\tilde W, \normF{\tilde W} = 1}\normF{\bigl(\tilde W C_X\trasp{\tilde W}\bigr)^{1/3}}^2$, \label{item1:sm_multiscale_problem_1}
\item solve the scalar minimization problem $\ds \min_{r>0} \Bigl(2r^2 - 3r^{4/3}\normF{\bigl(\tilde W^\star C_X\trasp{(\tilde W^\star)}\bigr)^{1/3}}^2\Bigr)$.
\label{item2:sm_multiscale_problem_1}
\end{enumerate}
We begin by analyzing item~\ref{item1:sm_multiscale_problem_1}. We have
\begin{align}
\normF{\bigl(\tilde W C_X\trasp{\tilde W}\bigr)^{1/3}}^2 &= \normF{\bigl(U_{\tilde W} \Sigma_{\tilde W} V_{\tilde W}^\top {V}_{C} {\Lambda}_{C} {V}_{C}^\top V_{\tilde W} \Sigma_{\tilde W}^\top U_{\tilde W}^\top\bigr)^{1/3}}^2 = \normF{\bigl(\Sigma_{\tilde W} V_{\tilde W}^\top {V}_{C} {\Lambda}_{C} {V}_{C}^\top V_{\tilde W} \Sigma_{\tilde W}^\top\bigr)^{1/3}}^2
\label{item:comp1}\\
&= \tr\bigl(\bigl(\Sigma_{\tilde W} V_{\tilde W}^\top {V}_{C} {\Lambda}_{C} {V}_{C}^\top V_{\tilde W} \Sigma_{\tilde W}^\top\bigr)^{2/3}\bigr),
\label{item:comp3}
\end{align}
where in equality~\eqref{item:comp1} we substituted the SVD of the matrices $\tilde W$ and $C_X$ and used the orthogonality of $U_{\tilde W}$, and in equality~\eqref{item:comp3} we used the symmetry of the matrix $\Sigma_{\tilde W} (V_{\tilde W})^\top {V}_{C} {\Lambda}_{C} ({V}_{C})^\top V_{\tilde W} \Sigma_{\tilde W}$\footnote{$\bigl(\Sigma_{\tilde W} (V_{\tilde W})^\top {V}_{C} {\Lambda}_{C} {V}_{C}^\top V_{\tilde W} \Sigma_{\tilde W} \bigr)^\top = \Sigma_{\tilde W} V_{\tilde W}^\top {V}_{C} {\Lambda}_{C} ({V}_{C})^\top V_{\tilde W} \Sigma_{\tilde W}$.}.
Now, note that 
\begin{enumerate}
\item the matrix $V_{\tilde W}^\top {V}_{C} \in \R^{n \times n}$ is orthogonal. In fact, we have
$
V_{\tilde W}^\top {V}_{C} \bigl(V_{\tilde W}^\top {V}_{C}\bigr)^\top = V_{\tilde W}^\top {V}_{C} {V}_{C}^\top V_{\tilde W} = I_n.
$
Since $\Sigma_{\tilde W} \in \R^{m \times n}$ and $\Lambda_{C}\in\R^{n \times n}$ are diagonal, Lemma~\ref{lem:pita} implies that the maximum of~\eqref{item:comp3} occurs when $V_{\tilde W}^\top {V}_{C} = {V}_{C}^\top V_{\tilde W} = I_n$.
\item 
$
\Sigma_{\tilde W} {\Lambda}_{C} \Sigma_{\tilde W}^\top =
\left[
\begin{array}{c|c}
\Sigma_{\tilde W}^m & 0_{m, n-m}
\end{array}
\right]
\left[
\begin{array}{c|c}
{\Lambda}_{C}^m & 0_{m, n-m} \\
\hline
0_{n-m, m} & {\Lambda}_{C}^{n-m}
\end{array}
\right]
\left[
\begin{array}{c}
\Sigma_{\tilde W}^m \\
\hline 
0_{n-m, m}
\end{array}
\right] = \Sigma_{\tilde W}^m {\Lambda}_{C}^m \Sigma_{\tilde W}^m \in \R^{m \times m}.
$
\item 
$ \normF{\tilde W} = 1 \iff  \tr{\bigl(\tilde W \tilde W^\top\bigr)} = 1 \iff \tr{\bigl(U_{\tilde W} \Sigma_{\tilde W} V_{\tilde W}^\top V_{\tilde W} \Sigma_{\tilde W}^\top U_{\tilde W}^\top\bigr)} \iff \tr{\bigl(\Sigma_{\tilde W} \Sigma_{\tilde W}^\top\bigr)} = \normF{\Sigma_{\tilde W}^m}^2 = 1.$
\end{enumerate}

Therefore, the OP in item~\ref{item1:sm_multiscale_problem_1} is equivalent to the following maximization problem
\begin{align}
\ds \max_{\tilde W, \normF{\tilde W} = 1}\normF{\bigl(\tilde W C_X\trasp{\tilde W}\bigr)^{1/3}}^2
&= \max_{\Sigma_{\tilde W}^m, \normF{\Sigma_{\tilde W}^m} = 1} \sum_{i = 1}^{m}\bigl( \lambda_i\bigl(\Sigma_{\tilde W}^m {\Lambda}_{C}^m \Sigma_{\tilde W}^m \bigr)\bigr)^{2/3} = \max_{\sigma^{\tilde W} \in \R^m_{>0}, \norm{\sigma^{\tilde W}}_2 = 1} \sum_{i=1}^m (\sigma_i^{\tilde W})^{4/3} (\lambda_i^C)^{2/3}\nonumber\\
&= \max_{\sigma^{\tilde W} \in \R^m_{>0}, \norm{\sigma^{\tilde W}}_2 = 1} \Bigl((\sigma^{\tilde W})^{4/3}\Bigr)^\top (\lambda^C)^{2/3}. \label{eq:max_w_final}
\end{align}
By applying H\"older inequality (see, e.g.~\cite[Exercise 2.4]{FB:24-CTDS}) to the scalar product $\Bigl((\sigma^{\tilde W})^{4/3}\Bigr)^\top (\lambda^C)^{2/3}$ with $x = (\sigma^{\tilde W})^{4/3}$, $y = (\lambda^C)^{2/3} \in \R^m$, $p = 3/2$ and $q = 3$, we have
$$
\Bigl((\sigma^{\tilde W})^{4/3}\Bigr)^\top (\lambda^C)^{2/3} \leq \norm{(\sigma^{\tilde W})^{4/3}}_p \norm{(\lambda^C)^{2/3}}_q,
$$
with equality if and only if there exists a constant $h > 0$ such that $h ((\sigma_i^{\tilde W})^{4/3})^{(3/2)} = ((\lambda_i^C)^{2/3})^3$.
Therefore~\eqref{eq:max_w_final} is maximized at $h (\sigma_i^{\tilde W})^{2} = (\lambda_i^C)^{2}$.
By imposing the constraint $\norm{\sigma^{\tilde W}}_2 = 1$, we obtain $h = \norm{\lambda^C}_2$.

Summing up, we have proved that the OP in item~\ref{item1:sm_multiscale_problem_1} is maximized at $\tilde W^\star = U_{\tilde W} \frac{\Lambda_C^m}{\normF{\Lambda_C^m}}V_{\tilde W}^\top$, for any orthogonal matrix $U_{\tilde W} \in \R^{m \times m}$ and where $V_{\tilde W} = (v_C^1, \dots, v_C^m)$. This concludes the proof of step~\ref{item1:sm_multiscale_problem_1}.
Additionally, we note
\begin{align*}
\ds \normF{\bigl(\tilde W^\star C_X\trasp{(\tilde W^\star)}\bigr)^{1/3}}^2 = \Bignorm{\Bigl(U_{\tilde W} \frac{\bigl(\Lambda_C^m\bigr)^3}{\normF{\Lambda_C^m}^2}\trasp{U_{\tilde W}}\Bigr)^{1/3}}_{\textup{F}}^2 = \frac{1}{\normF{\Lambda_C^m}^{4/3}}\normF{\Lambda_C^m}^2 =\normF{\Lambda_C^m}^{2/3},
\end{align*}
where in the first equality we used the fact that $U_{\tilde W}$ is orthogonal.

Next, for step~\ref{item2:sm_multiscale_problem_1}, we solve the OP
\beq
\label{eq:problem_w_r}
\ds \min_{r>0} \Bigl(2r^2 - 3r^{4/3}\normF{\bigl(\tilde W^\star C_X\trasp{(\tilde W^\star)}\bigr)^{1/3}}^2\Bigr) = \ds \min_{r>0} \Bigl(2r^2 - 3r^{4/3}\normF{\Lambda_C^m}^{2/3}\Bigr).
\eeq

Let $g(r) := 2r^2 - 3r^{4/3}\normF{\Lambda_C^m}^{2/3}$ and note that 
$$
g'(r) = 0 
\iff r\bigl(1 - \normF{\Lambda_C^m}^{2/3}r^{-2/3} \bigr) = 0 
\iff r = \normF{\Lambda_C^m} := r^\star.
$$
Additionally 
$
g''(r^\star) = 4 - \frac{4}{3}\normF{\Lambda_C^m}^{2/3} (r^\star)^{-2/3} = 4 - \frac{4}{3} > 0.
$
Therefore we have
\[
r_w := \argmin \Bigl(2r^2 - 3r^{4/3}\normF{\Lambda_C^m}^{2/3}\Bigr) = \normF{\Lambda_C^m}.
\]

Finally, we can conclude that the minimization problem of the third level~\eqref{eq:op_layer3} is optimized at
$$
W^\star = r_w \tilde W^\star = \normF{\Lambda_C^m} U_{\tilde W} \frac{\Lambda_C^m}{\normF{\Lambda_C^m}}V_{\tilde W}^\top = U_{\tilde W} \Lambda_C^m V_{\tilde W}^\top.
$$
This concludes the proof.
\end{proof}
\begin{rem}[Rayleigh Quotient]
When $m=1$, the maximization problem in item~\ref{item1:sm_multiscale_problem_1} simplifies to a well-known problem called \emph{Rayleigh quotient}. Specifically, we have:
\begin{align}
\label{eq:rayleigh_quotient}
\ds \max_{\tilde W\in\R^{1\times m}, \normF{\tilde W} = 1}\normF{\bigl(\tilde W C_X\trasp{\tilde W}\bigr)^{1/3}}^2 &= \Bigl(\max_{\tilde W\in\R^{1\times m}, \lambda_1^W = 1} {\tilde W C_X\trasp{\tilde W}}\Bigr)^{1/3}.
\end{align}
The maximization problem $\ds \max_{\tilde W\in\R^{1\times m}, \lambda_1^W = 1}{\tilde W C_X\trasp{\tilde W}}$ is known as the Rayleigh quotient.
It is well-known (see, e.g.,~\cite{RAH-CRJ:12}) that the maximum value of the Rayleigh quotient is the largest eigenvalue $\lambda_1^C$ of $C_X$. Moreover, this maximum is achieved when $\tilde W$ is the corresponding eigenvector $v_1^C$, that is $\tilde W^\star = v_1^C$.
Substituting this result back into problem~\eqref{eq:rayleigh_quotient}, we get 
$$
\max_{\tilde W\in\R^{1\times m}, \normF{\tilde W} = 1} {(W C_X\trasp{\tilde W})^{1/3}} =  \bigl(\lambda_1^C\bigr)^{1/3}.
$$
\end{rem}
\bibliographystyle{plainurl+isbn}
\bibliography{alias,FB,Main,new}
\end{document}